\documentclass[onecolumn]{IEEEtran}
\usepackage{amsmath, amssymb, amscd, tikz,tikz-cd}
\usepackage{amsthm}
\usepackage{times,amssymb,amsmath,amsfonts,float,nicefrac,color,bbm,mathrsfs,caption,float}
\usepackage{algorithm,enumerate,graphicx}
\usepackage[font={small}]{caption}
\usepackage[mathscr]{eucal}
\usetikzlibrary{matrix}
\usetikzlibrary{shapes.misc, positioning}
\usepackage{sidecap}
\usepackage{algpseudocode}
\usepackage{verbatim}
\usepackage{epstopdf}
\usepackage{textcomp}
\usetikzlibrary{shapes,arrows}
\usepackage{stmaryrd}
\usepackage{mathabx}
\usepackage[noadjust]{cite}
\usepackage{booktabs}
\usepackage{enumitem}
\setlist[enumerate]{leftmargin=*}
\usepackage[normalem]{ulem}
\newcommand{\blue}[1]{{\color{black}#1}}

\newcommand{\ff}{{\mathbb{F}}}
\newcommand{\PP}{{\mathbb{P}}}

\newcommand{\kk}{{\mathbbm k}}

\newtheorem{theorem}{Theorem}[section]

\newtheorem{proposition}[theorem]{Proposition}
\newtheorem{definition}[theorem]{Definition}
\newtheorem{example}[theorem]{Example}

\newcommand*{\QEDA}{\hfill\ensuremath{\blacksquare}}

\newtheorem{remark}[theorem]{Remark}

\newcommand{\RN}[1]{%
  \textup{\expandafter{\romannumeral#1}}%
}

\newcommand\remove[1]{}
\newcommand{\nc}{\newcommand}
\SetSymbolFont{stmry}{bold}{U}{stmry}{b}{n}

\allowdisplaybreaks
\nc\bfa{{\boldsymbol a}}\nc\bfA{{\boldsymbol A}}\nc\cA{{\mathcal A}}\nc\sA{{\mathscr A}}
\nc\bfb{{\boldsymbol b}}\nc\bfB{{\boldsymbol B}}\nc\cB{{\mathcal B}}\nc\sB{{\mathscr B}}
\nc\bfc{{\boldsymbol c}}\nc\bfC{{\boldsymbol C}}\nc\cC{{\mathcal C}}\nc\sC{{\mathscr C}}
\nc\bfd{{\boldsymbol d}}\nc\bfD{{\boldsymbol D}}\nc\cD{{\mathcal D}}\nc\sD{{\mathscr D}}
\nc\bfe{{\boldsymbol e}}\nc\bfE{{\boldsymbol E}}\nc\cE{{\mathcal E}}
\nc\bff{{\boldsymbol f}}\nc\bfF{{\boldsymbol F}}\nc\cF{{\mathcal F}}\nc\sF{{\mathscr F}}
\nc\bfg{{\boldsymbol g}}\nc\bfG{{\boldsymbol G}}\nc\cG{{\mathcal G}}
\nc\bfh{{\boldsymbol h}}\nc\bfH{{\boldsymbol H}}\nc\cH{{\mathcal H}}
\nc\bfi{{\boldsymbol i}}\nc\bfI{{\boldsymbol I}}\nc\cI{{\mathcal I}}\nc\sI{{\mathscr I}}
\nc\bfj{{\boldsymbol j}}\nc\bfJ{{\boldsymbol J}}\nc\cJ{{\mathcal J}}
\nc\bfk{{\boldsymbol k}}\nc\bfK{{\boldsymbol K}}\nc\cK{{\mathcal K}}
\nc\bfl{{\boldsymbol l}}\nc\bfL{{\boldsymbol L}}\nc\cL{{\mathcal L}}
\nc\bfm{{\boldsymbol m}}\nc\bfM{{\boldsymbol M}}\nc\cM{{\mathcal M}}
\nc\bfn{{\boldsymbol n}}\nc\bfN{{\boldsymbol N}}\nc\cN{{\mathcal N}}
\nc\bfo{{\boldsymbol o}}\nc\bfO{{\boldsymbol O}}\nc\cO{{\mathcal O}}
\nc\bfp{{\boldsymbol p}}\nc\bfP{{\boldsymbol P}}\nc\cP{{\mathcal P}}\nc\eP{{\EuScriptP}}\nc\fP{{\mathfrak P}}
\nc\bfq{{\boldsymbol q}}\nc\bfQ{{\boldsymbol Q}}\nc\cQ{{\mathcal Q}}
\nc\bfr{{\boldsymbol r}}\nc\bfR{{\boldsymbol R}}\nc\cR{{\mathcal R}}\nc\sR{{\mathscr R}}
\nc\bfs{{\boldsymbol s}}\nc\bfS{{\boldsymbol S}}\nc\cS{{\mathcal S}}
\nc\bft{{\boldsymbol t}}\nc\bfT{{\boldsymbol T}}\nc\cT{{\mathcal T}}
\nc\bfu{{\boldsymbol u}}\nc\bfU{{\boldsymbol U}}\nc\cU{{\mathcal U}}
\nc\bfv{{\boldsymbol v}}\nc\bfV{{\boldsymbol V}}\nc\cV{{\mathcal V}}\nc\sV{{\mathscr V}}
\nc\bfw{{\boldsymbol w}}\nc\bfW{{\boldsymbol W}}\nc\cW{{\mathcal W}}\nc\sW{{\mathscr W}}
\nc\bfx{{\boldsymbol x}}\nc\bfX{{\boldsymbol X}}\nc\cX{{\mathcal X}}
\nc\bfy{{\boldsymbol y}}\nc\bfY{{\boldsymbol Y}}\nc\cY{{\mathcal Y}}
\nc\bfz{{\boldsymbol z}}\nc\bfZ{{\boldsymbol Z}}\nc\cZ{{\mathcal Z}}
\nc{\integers}{\mathbb{Z}}

\DeclareMathOperator{\supp}{supp}

\begin{document}
\title%[Codes with hierarchical locality]
{Codes with hierarchical locality from covering maps of curves}

\author{\IEEEauthorblockN{Sean Ballentine$^{\dag}$}\quad
\and \IEEEauthorblockN{Alexander Barg$^{\ddag}$}\quad
\and \IEEEauthorblockN{Serge Vl{\u a}du{\c t}$\,^{\ast}$}}
\maketitle
%\author[S. Ballentine, A. Barg, and S. Vl{\u a}du{\c t}]{{Sean Ballentine$^\dag$}\quad{Alexander Barg$^\ddag$}\quad{Serge Vl{\u a}du{\c T}$^\ast$}}

\renewcommand{\thefootnote}{}\footnotetext{

\vspace{-.1in}
 
%\noindent\rule{1.5in}{.4pt}

\noindent $^\dag$  Dept. of Mathematics, University of Maryland, College Park, MD 20742. Email: seanballentine@gmail.com. 

\noindent $^\ddag$ Dept. of ECE and ISR, University of Maryland, College Park, MD 20742 and IITP, Russian Academy of Sciences, Moscow, Russia. Email abarg@umd.edu. Research supported by NSF grants CCF1422955 and CCF1618603.

\noindent $^\ast$ Aix Marseille Universit{\'e}, CNRS,
Centrale Marseille,
I2M UMR 7373, 13453, Marseille, France, and IITP, Russian Academy of Sciences, Moscow, Russia. Email
serge.vladuts@univ-amu.fr.

}
\renewcommand{\thefootnote}{\arabic{footnote}}
\setcounter{footnote}{0}

\thispagestyle{empty}

\begin{abstract}
Locally recoverable (LRC) codes provide ways of recovering erased coordinates of the codeword without having to access each of the remaining coordinates. A subfamily of LRC codes with hierarchical locality (H-LRC codes) provides added flexibility to the construction by introducing several tiers of recoverability for correcting different numbers of erasures. We present a general construction of codes with 2-level hierarchical locality from maps between algebraic curves  and specialize it to several code families obtained from quotients of curves by a subgroup of the automorphism group, including rational, elliptic, Kummer, and Artin-Schreier curves. We further address the question of H-LRC codes with availability, and suggest a general construction of such codes from fiber products of curves. Detailed calculations of parameters for H-LRC codes with availability are performed for Reed-Solomon- and Hermitian-like code families. Finally, we construct asymptotically good families of H-LRC codes from curves related to the Garcia-Stichtenoth tower.
\end{abstract}

\section{introduction}
Locally recoverable (LRC) codes form a family of erasure codes motivated by applications in distributed storage that support
repair of a failed storage node by contacting a small number of other nodes in the cluster. While in most situations repairing a single failed node restores the system to the functional state, occasionally there may be a need to recover the data from
several concurrent node failures. Addressing this problem, several papers have constructed families of LRC codes that locally correct multiple erasures \cite{kamath2012codes,tamobarg}. In this paper we consider the intermediate situation when the code corrects a single erasure by 
contacting a small number $r_2$ of helper nodes, while at the same time supporting local recovery of multiple erasures. 
This gives rise to LRC codes with hierarchy, originally defined in \cite{hierarchical}. We observe that the hierarchical locality
property arises naturally in constructions of algebraic geometric LRC codes, leading to a general construction of such codes from
covering maps in towers of algebraic curves.

Paper \cite{hierarchical} obtained an upper bound on the distance of H-LRC codes in terms of the dimension and locality parameters.
Codes that meet this bound with equality are called (distance)-optimal. Optimal H-LRC codes with 2-level locality
over $\ff_q$ of length $n\le q-1$ were constructed in \cite{hierarchical}, expanding on the construction of Reed-Solomon subcodes in 
\cite{tamobarg}. Another generalization of the construction in \cite{tamobarg} builds upon a geometric view of these codes,
and expands it to codes obtained from covering maps of algebraic curves \cite{BTV17}. Using that approach, several follow-up
papers constructed a number of families of LRC codes on curves \cite{B+17,Li17a,LiMaXing17b,JinMaXing17,HMM}. In this paper we further
extend the basic construction of LRC codes on curves to construct LRC codes with hierarchy. Our main result is a general
construction of such codes from covering maps, and we use it to obtain families of H-LRC codes based on quotient curves and other
well-known towers of curves, including quotients of elliptic, Kummer, and Artin-Schreier curves. We also construct H-LRC codes of unbounded length from curves related to the Garcia-Stichtenoth tower \cite{garcia95}, observing that they yield an asymptotically good family of codes. Finally, we briefly consider H-LRC codes with multiple recovering sets, addressing the so-called 
availability problem \cite{RPDV16,tamobarg} in the hierarchical setting.

%
%
%
%
%The construction in \cite{hierarchical} was inspired by 
%\cite{tamobarg} as are our constructions in this paper. We also rely on the ideas from \cite{BTV17,JinMaXing17}, which enables
%us to construct optimal 2-level H-LRC codes of length up to $q+1.$ Further, we construct 2-level H-LRC codes from 
%quotients of elliptic curves, following in the footsteps of \cite{LiMaXing17b} which used them constructed LRC codes without hierarchy.
%This enables us to increase the code length $n$ to about $q+2\sqrt q.$ We present constructions of H-LRC codes from quotients of 
%Kummer and Artin-Schreier curves, and 

A preliminary version of this work was presented at the 2018 IEEE International Symposium on Information Theory \cite{Bal18}.
At the same time, most of the material in Sections~\ref{sect:elliptic}-\ref{sect:GS} was not included in \cite{Bal18} and appears here for the first time.
%The present paper contains a number of new results compared to , most notably, %\ref{sect:examples},
%the constructions in Sections~\ref{sect:elliptic}-\ref{sect:GS}.

\section{Definitions}
A code is LRC if every coordinate of the codeword is a function of only a small number of other coordinates. 
Formalizing this concept, we obtain the following definition.

\begin{definition}[LRC codes, \cite{gopalan2012locality}]\label{def:LRC} 
A code $\cC\subset {\ff}_q^n$ is \emph{locally recoverable with locality $r$}
if for every $i\in \{1,2,\dots,n\}$ there exists an $r$-element subset 
$I_i\subset \{1,2,\dots,n\}\backslash \{i\}$
and a function $\phi_i:{\ff}_q^r\to {\ff}_q$ such that for every codeword $x\in\cC$ we have
   \begin{equation}\label{eq:def1}
   x_i=\phi_i(x_{j_1},\dots,x_{j_r}),
   \end{equation}
where $j_1 < j_2 < \cdots < j_r$ are the elements of $I_i.$
%This definition can be rephrased as follows. Given $a\in {\ff}_q,$ consider the sets of codewords
%   \begin{equation*}
%   \cC(i,a)=\{x\in \cC: x_i=a\},\quad i\in\{1,2,\dots,n\}.
%   \end{equation*}
%The code $\cC$ is said to be \emph{locally recoverable with locality $r$} if for every $i$
%there exists an $r$-element subset $I_i\subset \{1,2,\dots,n\}\backslash \{i\}$
%such that the restrictions of the sets $\cC(i,a)$ to
%the coordinates in $I_i$ for different $a$ are disjoint:
%   \begin{equation}\label{eq:cap}
% \cC_{I_i}(i,a)\cap \cC_{I_i}(i,a')=\emptyset,\quad  \text{whenever $a\ne a'$}.
%   \end{equation}
\end{definition}
For a given coordinate $i\in \{1,\dots,n\}$ the set $I_i$ is called the {\em recovering set} of $i$. We denote the restriction 
$\cC|_{\{i\}\cup I_i}$ of the code $\cC$ to the coordinates in $\{i\}\cup I_i$ by $\cC_i,$ and we call the set $\{i\}\cup I_i$ a {\em repair group}. Note that the length of $\cC_i$ is $r+1.$
%The rephrasing of \eqref{eq:def1} given in \eqref{eq:cap} does away with the functions $\phi_i$ by replacing them with a partition of $\cC_i$ into subsets of vectors with the same $x_i$. 

In this paper we study only linear LRC codes. For them the above definition can be phrased as follows: For every $i\in\{1,2,\dots,n\}$ there exists a punctured code $\cC_i:=\cC|_{\{i\}\cup I_i}$ such that $\dim(\cC_i)\le r$ and distance $d(\cC_i)\ge 2.$ 
Since $\cC_i$ corrects one erasure, every coordinate in the repair group $\{i\}\cup I_i$ can be locally recovered.

In this form, the definition of LRC codes is easily extended to local correction of more than one erasure. Following \cite{kamath2012codes}, we say that a linear code has locality $(r,\rho)$ if for every $i\in \{1,2,\dots,n\}$ there exists a 
subset $I_i\subset\{1,\dots,n\}\backslash\{i\}$ such that the code $\cC_i=\cC|_{i\cup I_i}$ has dimension $\dim(\cC_i)\le r$ and distance $d(\cC_i)\ge \rho.$ In this case any $\rho-1$ erasures can be locally corrected, and we again refer to the set $\{i\}\cup I$
as a repair group. Although this is not needed in this definition, earlier works assumed that $|I_i|=r+\rho$, and that $\dim(\cC_i)=r,d(\cC_i)=\rho+1$, i.e., that the code $\cC_i$ is maximum distance separable (MDS); see for instance \cite{tamobarg,BTV17}.

Let $\cC$ be an LRC code of length $n$, cardinality $q^k$, and distance $d$ (briefly, an $[n,k,d]$ code) with locality $(r,\rho).$
The minimum distance of $\cC$ is bounded above as follows \cite{kamath2012codes}:
   \begin{equation}\label{eq:sb}
   d\le n-k+1-\Big(\Big\lceil\frac kr\Big\rceil-1\Big)(\rho-1).
   \end{equation}
 For correction of a single erasure, this result reduces to the Singleton-type bound of \cite{gopalan2012locality}:
 The distance of an LRC code with locality $r$ is bounded above as 
  \begin{equation}\label{eq:sb2}
  d\le n-k+2-\lceil k/r\rceil.
  \end{equation} 
We say that a code with locality $(r,\rho)$ is {\em optimal} if its parameters meet the bound \eqref{eq:sb}-\eqref{eq:sb2} with equality.
      
In the following definition, due to \cite{hierarchical}, we introduce linear codes with hierarchical locality, which form the main subject of our paper.
\begin{definition}[H-LRC codes \cite{hierarchical}]\label{def:hie} Let $\rho_2< \rho_1$ and $r_2\le r_1.$ A linear code $\cC$ is H-LRC and 
parameters $((r_1,\rho_1),(r_2,\rho_2))$ if for every $i\in\{1,\dots,n\}$ there is a punctured code $\cC_i$ such that {$i\in\supp(\cC_i)$} and
   \begin{enumerate}
   \item $\dim(\cC_i)\le r_1$,
    \item $d(\cC_i)\ge \rho_1,$ and 
    \item $\cC_i$ is an $(r_2,\rho_2)$ LRC code.
    \end{enumerate}
%    We assume that $r_2<r_1$ and $\rho_2<\rho_1$.
\end{definition}
The intuition behind this definition is that any $\rho_2-1$ erasures can be recovered using the local correction procedure
of the code $\cC_i$ (i.e., using recovering sets of size $r_2$ within the support $\supp(\cC_i)$),
and any larger number of erasures up to $\rho_1-1$ can be recovered using the entire set of coordinates
of the code $\cC_i$. Below we call the codes $\cC_i$ the {\em middle codes} and denote their length by $\nu$. {Thus, 
$\cC_i$ is a $[\nu,r_1,\rho_1]$ LRC code with locality $r_2$, and in all our constructions $\rho_2=2$ which corresponds to local 
correction of a single erasure (but see Proposition~\ref{prop:many} and the related discussion).} In all of our constructions
the coordinate set $[n]$ will be partitioned into disjoint groups of size $\nu$, and thus, the codes $\cC_i$ coincide for all
$i$ within each of the groups, and have disjoint supports otherwise. For the purposes of this paper, we could incorporate this property 
into the definition of the H-LRC code.

This definition can be extended by induction to any number of levels of hierarchy in an obvious way, and we denote the set of
parameters of a $\tau$-level H-LRC code by $(r_i,\rho_i), i=1,\dots,\tau.$ 
A bound on the distance of a $\tau$-level H-LRC code that extends \eqref{eq:sb} to all $\tau\ge 1$, takes the following form \cite{hierarchical}: 
   \begin{equation}\label{eq:sh}
     d\le n-k+1-\Big(\Big\lceil\frac {k}{r_\tau}\Big\rceil-1\Big)(\rho_\tau-1)-
     \sum_{j=1}^{\tau-1}\Big(\Big\lceil\frac{k} {r_j}\Big\rceil-1\Big)(\rho_j-\rho_{j+1}).
   \end{equation}
   An H-LRC code whose parameters meet this bound with equality will be called {\em optimal} throughout.
   
In this work we extend constructions of optimal LRC codes in the sense of \eqref{eq:sb}-\eqref{eq:sb2} to the hierarchical case.
There are several constructions of optimal LRC codes in the literature \cite{kamath2012codes,sil13,TPG13,tamobarg,MaGe17,Huang17,LiuMesnagerChen18}. Among them we single out the
construction of \cite{tamobarg} which isolates certain subcodes of Reed-Solomon (RS) codes that have the locality property. This
code family relies on an algebraic structure of LRC codes that affords an extension to codes on algebraic curves. The theory of algebraic geometric codes with locality, introduced in \cite{BTV17} and further
developed in \cite{B+17,LiMaXing17b,JinMaXing17} provides a framework for our study here, and we describe it in the next section.

\section{LRC codes on algebraic curves}\label{sec:AGLRC}
The following construction of LRC codes from covering maps of algebraic curves was introduced in \cite{BTV17} and is based on the approach in \cite{tamobarg} (even though the authors of \cite{tamobarg} did not phrase their results in geometric language). Let $\phi:X\to Y$ be a rational separable map of smooth projective absolutely irreducible curves of degree $r+1$ over a finite field 
$\kk$ and let $\phi^\ast:\kk(Y)\to \kk(X)$ be the corresponding map of the function fields. Since $\phi$ is separable, the primitive element theorem implies that there exists a function $x \in \kk(X)$ such that $\kk(X) = \kk(Y)(x)$.
Let $S=\{P_1, \dots, P_m\}$ be a set of $\kk$-rational points on $Y$ and let $Q_{\infty}$ be a positive divisor whose support is disjoint from $S$ (typically we choose $\text{supp}(Q_\infty)\subset \pi^{-1}(\infty)$ for a projection $\pi:Y\to \PP_{\kk}^1$). For each $i$, let $\{P_{ij}\}$ be the collection of points on $X$ in the preimage of $P_i$, i.e. $\{P_{ij}\}=\phi^{-1}(P_i)$. We assume that each
$P_i$ splits completely in the function field $\kk(X)$, and therefore $|\phi^{-1}(P_i)|=r+1$ for some fixed integer $r$ and all $i=1,\dots,m.$ Finally, define the set of points 
   $$
   D=\bigcup_{i=1}^m \bigcup_{j=1}^{r+1} P_{ij}\subset X(\kk)
   $$
    that will serve the
evaluation points of the code that we are constructing.
%These points will be the evaluation points for a locally recoverable code.

%Pick a positive integer $t$ and let $\{f_1, \dots, f_m\}$ be a basis of the linear space $L(tQ_{\infty})$.
Let $\{f_1, \dots, f_t\}$ be a basis of the linear space $L(Q_{\infty}),$  where $t:=\dim(L(Q_{\infty})).$ 
These functions can be thought of as functions in $\kk(X)$ by the embedding of function fields $\phi^\ast,$ and each of these functions is constant on the fibers of $\phi$. Let $V$ be the subspace of $\kk(X)$ of dimension $rt$ spanned over $\kk$ by the functions 
\begin{equation}\label{eq:basis}
\{f_jx^i, i=0, \dots, r-1, j=1, \dots, t\} \text{.}
\end{equation}
The code $\cC(D, \phi)$  is defined as the image of the map
\begin{equation}\label{eq:map}
\begin{aligned}\text{ev}_D : &V \to \kk^{(r+1)m}\\
&v \mapsto (v(P_{ij}), i = 1, \dots, m, j=1, \dots, r+1).
\end{aligned}
\end{equation}

The code $\cC(D, \phi)$ is locally recoverable with repair groups of size $r+1.$ Denote by $c_{ij}$ the position in the codeword that corresponds to the point $P_{ij}$. The recovering set for $c_{ij}$ is formed by the $r$ positions given by the points 
$\{P_{il},l\ne j\}.$ Recovery of $c_{ij}$ proceeds by polynomial interpolation.
Properties of the codes generated by this construction are well-studied and for more information about their parameters and basic examples of such codes we once again refer the reader to \cite{BTV17}. %which are $\{c_{ik}\}_{k\ne j} $. Recovery works by polynomial interpolation in the following way. On a particular fiber, functions in $L(tQ_\infty)$ are constant. If your erasure occurs at position $c_{ij}$, then restricted to the fiber $\phi^{-1}(P_j)$, a function $g \in V$ takes the form
%$$g=\sum_{i=0}^{r-1}\alpha_{ij} x^i$$
%which is a polynomial in $x$ of degree at most $r-1$. The codeword $\Delta(g)$ can therefore be recovered at position $c_{ij}$ by noting the values of $\Delta(g)$ at the other $r$ positions $\{ c_{ik} | \, k \neq j\}$.

\subsection{The case of the projective line and the construction of \cite{tamobarg}}\label{sect:TBcodes} A particular case of this construction that arises when both $X$ and $Y$ are taken to be projective lines $\PP^1_{\kk},$ gives
rise to the RS-type LRC codes constructed in \cite{tamobarg}. These codes are constructed as evaluations of 
functions from the $k$-dimensional linear space $V\subset \ff_q[x]$ spanned by $\{\phi^jx^i, j=0,1,\dots,k/r-1, i=0,1,\dots,r-1\},$ where
$\phi\in\ff_q[x]$ is a polynomial of degree $r+1$ that is constant on the repair groups of size $r+1.$ As in \cite{tamobarg}, let us assume that $(r+1)|n$ and $r|k.$ Applying the general definition 
\eqref{eq:basis}-\eqref{eq:map} in this case, we obtain optimal LRC codes whose parameters meet the bound in \eqref{eq:sb2} with equality. Indeed, the maximum degree of a polynomial in $V$ is $(\frac{k}{r}-1)(r+1) + (r-1)=k\frac{r+1}{r}-2,$ and therefore, $d_{\text{min}}(\cC)\ge n-k\frac{r+1}{r}+2$.

Moreover, increasing the degree of $\phi$ from $r+1$ to $r+\rho-1,\rho\ge 2$ and using the same
construction as above with repair groups of size $r+\rho-1$, we obtain a class of LRC codes whose repair groups are resilient to up to $\rho-1$ erasures. For a chosen value of $\rho\ge 2$ and for $(r+\rho-1)|n$ the parameters of
these codes meet the bound in \eqref{eq:sb} with equality. %This code family is discussed in detail in \cite{tamobarg}.

\section{H-LRC codes on algebraic curves}\label{sec:curves}

In this section we present a natural extension of the construction from the previous section that gives rise to LRC codes with hierarchy. Let 
$X,Y,$ and $Z$ be smooth projective absolutely irreducible curves over a finite field 
$\kk$. Consider the following sequence of maps:
\begin{equation}\label{eq:tow}
X\xrightarrow{\phi_2}Y\xrightarrow{\phi_1}Z,
\end{equation}
where $\phi_1$ and $\phi_2$ are rational seperable maps of degree $s+1$ and $r_2+1,$ respectively, where $s\ge 2, r_2\ge 1.$ Define $\psi:=\phi_1\circ\phi_2.$
Let $\phi_2^\ast:\kk(Y)\to \kk(X)$ and $\phi_1^\ast:\kk(Z)\to\kk(Y)$ be the 
corresponding maps of the function fields. Let $x\in \kk(X)$ and $y
\in\kk(Y)$ be primitive elements of their respective algebraic extensions, i.e., suppose that 
$\kk(X) = \kk(Y)(x)$ and $\kk(Y) = \kk(Z)(y)$. Let $S=\{P_1,\dots, P_m\}$ be a collection of points on $Z(\kk)$ 
that split completely on $X$, i.e., $|\psi^{-1}(P_i)| = (r_2+1)(s+1)$. Let $D=
\bigcup_{i=1}^{m}\psi^{-1}(P_i), n:=|D|,$ and let $Q_{\infty}$ be a positive divisor on $Z$ with support disjoint from $S$. We will assume that $\text{supp}((y)_\infty)\cap \phi_2^{-1}(S)=\emptyset$ and $\text{supp}((x)_\infty)\cap \psi^{-1}(S)=\emptyset,$ where $(\cdot)_\infty$ is the polar divisor.

As before, let $\{f_1, \cdots, f_t\}$ be a basis for the space $L(Q_{\infty}).$
Let $V$ be the vector space of functions over $\kk$ spanned by 
  \begin{equation}\label{eq:V}
\{f_iy^jx^k| 1\le i \le t, 0\le j \le s-1, 0\le k \le r_2-1\}.
  \end{equation}
Let $\nu:=(s+1)(r_2+1)$ { and note that $n=m\nu$.} As in \eqref{eq:map}, define the code $\cC(D,\{\phi_1, \phi_2\})$  as the image of the evaluation map 
\begin{equation}\label{eq:hmap}
\begin{aligned}
\text{ev}_D : &V \to \kk^{n}\\
&v \mapsto (v(P), P\in D) \text{.}
\end{aligned}
\end{equation}
\blue{Below in Sec.~\ref{sec:RS} we give a simple example of the above construction, taking $X,Y,$ and $Z$ to be projective lines 
${\mathbb P}^1$ and constructing $\cC$ as a subcode of an RS code with hierarchical locality. We also illustrate erasure recovery by the local and middle codes. Note that the supports of the middle codes are formed by the preimages of the points in $S$ on the curve $X$. This is illustrated
in Fig.~\ref{fig:cover}.}

%\vspace*{.5in}
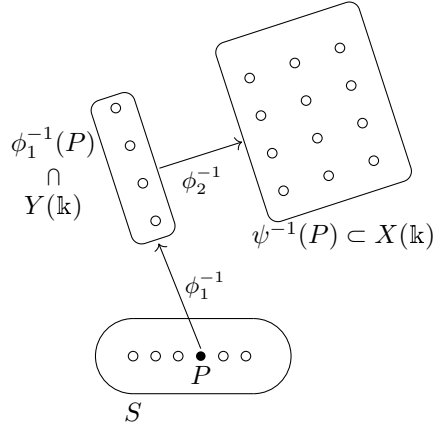
\begin{figure}[H]
\captionsetup{width=0.8\textwidth}
\begin{center}\begin{tikzpicture}

%set S on the curve Z
\draw(0,0) circle(.025in);
\fill[black!100](0.3,0) circle(.025in) node[anchor=north] {$P$};
\draw(-0.3,0) circle(.025in) ;
\draw(-0.6,0) circle(.025in);
\draw(0.6,0) circle(.025in);
\draw(0.9,0) circle(.025in);
\draw (-.6,0.5) arc (90:270:.5);
\draw (1,.5) arc (90:-90:.5);
\draw (-.6,.5) -- ++(0+0:1.6);
\draw  (1,-.5) -- ++(0+180:1.6) node[below,xshift=0cm] {$S$};

% fiber \phi_1^{-1}(P)
\draw(-.3,1.8) circle(.025in);
\draw(-.47,2.3) circle(.025in);
\draw(-.64,2.8) circle(.025in);
\draw(-.83,3.3) circle(.025in);
%\draw[rounded corners,rotate around={80:(-0.5,2)}] (-0.5,2.5) rectangle (1,-0.2);
\draw[rounded corners,rotate around={108:(-.6,2.5)}] (-1.6,2.2) rectangle (0.4,2.8)
node[pos=0, anchor=south] {\hspace*{-1.3in}$\text{$\begin{array}{c}\phi_1^{-1}(P)\\[-.08in]\rotatebox{270}{$\subset$}\\[.08in] Y(\kk)\end{array}$}$};

\draw[thin,<-] (-.26,1.5) -- (0.3,0.1) node[pos=0.4, anchor=west] {$\text{\small$\phi_1^{-1}$}$};

% fiber \psi^{-1}(P)

\draw(1.4,2.2) circle(.025in);
\draw(2.0,2.4) circle(.025in);
\draw(2.6,2.6) circle(.025in);
\draw(1.25,2.7) circle(.025in);
\draw(1.85,2.9) circle(.025in);
\draw(2.45,3.1) circle(.025in);
\draw(1.1,3.2) circle(.025in);
\draw(1.7,3.4) circle(.025in);
\draw(2.3,3.6) circle(.025in);
\draw(0.95,3.7) circle(.025in);
\draw(1.55,3.9) circle(.025in);
\draw(2.15,4.1) circle(.025in);

\draw[rounded corners,rotate around={18:(1.8,3.25)}] (0.8,2) rectangle (2.8,4.5);
\draw(2.2,1.6) node[] {$\text{$\psi^{-1}(P)\subset X(\kk)$}$};

\draw[thin,->] (-.25,2.5) -- (.85,2.85) node[pos=0.4, anchor=north] {\hspace*{.1in}$\text{\small$\phi_2^{-1}$}$};

\end{tikzpicture}
\end{center}
\caption{\blue{The point $P\in Z(\kk)$ is lifted to the curve $X$. The preimage $\psi^{-1}(P)$ forms the support of a middle code $\cC_\alpha.$ This code is LRC with locality $(r_2,2)$, and its repair groups are formed by the fibers of the covering $X\xrightarrow{\phi_2}Y.$ Univariate interpolation accounts for local repair of a single erasure by accessing the coordinates within its fiber, while bivariate interpolation over the entire set
of nonerased locations in $\psi^{-1}(P)$ recovers up to $\rho_1-1$ erasures. The example below in this section shows detailed calculations for
the case of projective lines.}}\label{fig:cover}
\end{figure}

Let $\deg(x)$ and $\deg(y)$ be
the degrees of the maps $x:X\to{\mathbb P}^1$ and $y:Y\to{\mathbb P}^1.$
Recall that \cite{BTV17} assumed that the function $x$ is injective on the fibers $\{P_{ij},j=1,\dots,r+1\}$ (see Sec.~\ref{sec:AGLRC}), and that this assumption holds in all the examples considered there. In our setting here, $x$ may not be injective on fibers of the map $\psi:=\phi_1\circ\phi_2$. Denote by $\deg_{\psi}(x)$ the largest number of zeros of the function $x:X\to{\mathbb P}^1$ on any single fiber of $\psi.$

\begin{proposition}\label{prop:general} The code $\cC=\cC(D,\{\phi_1, \phi_2\})$ is a 2-level H-LRC code {of length $n=m\nu$} with parameters
$ ((r_1, \rho_1),(r_2,\rho_2=2)) $, where the middle codes are of length $\nu=(s+1)(r_2+1),$ dimension $r_1=r_2s,$ and distance
 \begin{gather}\label{eq:rho1}
\rho_1\ge \max(2(r_2+1)-\deg_{\psi}(x)(r_2-1),4).
   \end{gather}
%   where $\nu:=(s+1)(r_2+1).$
  We also have 
   \begin{align}
   \dim(\cC)&=t r_2 s\ge r_1(\deg(Q_\infty)-g_Z+1) \label{eq:dim}\\
   d_{\text{\rm min}}(\cC)&\ge n-(\deg(Q_\infty)(s+1) +\deg (y)(s-1))(r_2+1)-\deg(x)(r_2-1), \label{eq:d}
   \end{align}
   where $g_Z$ is the genus of $Z$ and $t=\dim(L(Q_\infty)).$
\end{proposition}
{\em Proof:} The set $D$ of $n$ points is naturally partitioned into subsets of size $\nu,$ given by
the fibers of the covering map $\psi$ and each of them supports a code $\cC_\alpha,\alpha=1,\dots,n/\nu$ of length $\nu$.  The support of each of the
codes $\cC_\alpha$ is further partitioned into repair groups of size $r_2+1$ each of which is formed of the coordinates contained in
a particular fiber of the map $\phi_2$. 
Restricted to such a fiber, the functions $f_1, \dots , f_t$ and $y$ are constant, and any function in $V$ becomes a 
polynomial in $x$ of degree $\le r_2-1$. Therefore, $\cC$ restricted to a fiber of $\phi_2$  is an $r_2$-dimensional 
code with minimum distance $\rho_2$ determined by the maximum degree of such a polynomial in $x,$ which is $r_2-1$. 
The length of the restricted code is $r_2$, so it is a single parity check code with distance $\rho_2=2$.
Furthermore, this implies that each of the codes $\cC_\alpha$ (i.e., $\cC$ restricted to the fibers of $\psi$) is an LRC code with parameters $(r_2, 2)$. 

It remains to determine the parameters of the codes $\cC_\alpha$. First note that the functions $f_1, \dots ,f_t$ are constant on these fibers, and therefore, $V$ restricted to each of them becomes an $r_1$-dimensional space of functions spanned by 
    $$
    \{y^jx^k, j=0,1,\dots,s-1; k=0,1,\dots,r_2-1\}.
    $$ 
The minimum distance of $\cC_{\alpha}$ is determined by the maximum number of zeros of a non-zero function in $V,$ restricted to
a fiber of $\psi:$
  $$
\rho_1\ge \nu-(s-1)(r_2+1)-\deg_{\psi}(x)(r_2-1),
  $$
which gives the first term under the maximum in \eqref{eq:rho1}. \blue{Suppose that $\rho_1> 4.$ The code $\cC_\alpha$ can correct $\rho_1-1$ erasures by performing bivariate polynomial interpolation over some $r_1$ nonerased independent coordinates. 

If the first term in \eqref{eq:rho1} is trivial, we can show  $\rho_1\ge 4$ by proving that the code $\cC_\alpha$ corrects any three erasures. Indeed, if they are located in different repair groups of size $r_2+1,$ they can be recovered using the LRC properties of $\cC_\alpha.$
 Suppose that at least two of them fall in the same repair group. The restriction of the function $f$ \eqref{eq:hmap} to the support of the code
$\cC_\alpha$ is a bivariate polynomial with at most $r_1=r_2 s$ nonzero coefficients. On each of the remaining $s$ repair groups (fibers of $\phi_2$), the function $y$ is a constant, and we are left with a univariate polynomial of degree $r_2-1.$ Its coefficients can be recovered from $r_2$
independent evaluations on the fiber; thereby, we can recover the entire function $f$ restricted to the fiber of the mapping $\psi.$}
%These coefficients can be recovered by performing bivariate interpolation over $r_2 s$ points in the remaining repair groups (fibers of $\phi_2$), %each of which contains $r_2$ independent evaluations. }

 Finally, the bounds in \eqref{eq:dim}, \eqref{eq:d} are obtained by the same arguments applied to the code $\cC$ in its entirety. \QEDA

%\textcolor{red}{(The last part of the proof is somewhat imprecise. We have $(r_2+1)(s+1)$ coordinates of which one in each of the $s+1$
%fibers is dependent. Taking them out, we are left with $r_2(s+1)$ independent evaluations from which we are to recover $r_2s$ coefficients of
%the polynomial. Thus, we can correct}

\subsection{A family of optimal RS-like H-LRC codes}\label{sec:RS}
Using the above ideas, we show how the construction of RS-like codes in \cite{tamobarg} can be extended to yield
optimal two-level H-LRC codes. Let $\kk=\ff_q$ and let $r_2, r_1,$ and $n\le q$ be such that $r_1=sr_2,$ \blue{$(r_2+1)|\nu$,} and
$\nu|n$.

To construct the code we start with choosing a subset $D$ of $n$ points in $\kk$ and partition it into disjoint
subsets $D_\alpha$ of size $\nu$ each. Each of the subsets $D_\alpha$ will support an LRC code
of dimension $r_1$ and distance $r_2+3.$ The repair groups of this LRC code are of size $r_2+1.$ Assume that there is 
a polynomial $y\in \kk[x]$  of degree $r_2+1$ that is constant on these repair groups\footnote{A way to construct such polynomials is presented in \cite{tamobarg}, and we do not discuss it here.}. Further, 
we choose a polynomial $f\in \kk[x]$ of degree $\nu$ that is constant on each of the subsets 
$D_\alpha$.

For a positive integer $t$, let $V\subseteq\kk[x]$ be the $tr_1$-dimensional space spanned by 
   \begin{equation}\label{eq:V1}
   \{f^k y^j x^i,\;i=0,\dots,r_2-1,j=0,\dots,s-1,k=0,\dots,t-1\}.
   \end{equation}
 To connect this equation to \eqref{eq:V}, we note that the powers $f^k$ in \eqref{eq:V1} form a basis of the space $\cL((t-1)\infty)$ and  
 correspond to $f_i,i=1,\dots,t$ in \eqref{eq:V}. 
 %and $f_2$ corresponds to the function $y.$  
 Let us construct a code $\cC$ by evaluating these functions at the points in $D$ as described in \eqref{eq:hmap}.
For any $k$, the function $f^k$ is constant on each of the sets $D_\alpha,$ and therefore, the functions in $V$ restricted to each of these sets have degree at most 
$(s-1)(r_2+1)+r_2-1.$ 
This implies that the distance of $\cC_\alpha$ is at least
   \begin{align*}
   d_\text{min}(\cC_\alpha)&\ge \nu-(s-1)(r_2+1)-r_2+1\\
     &=r_2+3
   \end{align*}
which meets the bound \eqref{eq:sb2} with equality.

The dimension of the code $\cC$ is $\dim(V)=tr_1$ and the distance is found by counting the maximum degree of a function in $V,$
and is bounded below as
   \begin{equation}\label{eq:dred}
   d_{\text{min}}(\cC)\ge n-t(r_1+r_2+1+s)+r_2+3
   \end{equation}
 meeting the upper bound in \eqref{eq:sh}.
 We conclude with the following proposition.
 \begin{proposition}\label{prop:RS}
 Let $n\le q, t\ge 1$ and let $r_1,r_2$ be such that $r_1=sr_2$ for some $s>1$ and \blue{$(r_2+1)|\nu$}, $\nu|n.$
 The parameters of the code $\cC$ %constructed as in \eqref{eq:hmap} by evaluations of the polynomials in \eqref{eq:V}
  are $[n,tr_1,d_{\text{\rm min}}=n-t(r_1+r_2+1+s)+r_2+3].$ 
 Furthermore, $\cC$ is an optimal H-LRC code with two levels of hierarchy and locality parameters $(r_1,r_2+3),(r_2,2).$ 
 The middle codes $\cC_\alpha$ are optimal $[\nu,r_1,r_2+3]$ LRC codes.
 \end{proposition}
 The code family in this proposition is originally due to \cite{hierarchical}, where it was obtained as an extension of 
 \cite{tamobarg}, with no connection to the geometric interpretation. Making this connection enables us to increase the code
 length to $n=q+1$ in the next section.
 
 \vspace*{.1in}\blue{{\em Example:} The following example is very much in the spirit of the main construction of \cite{tamobarg}; see also Example 1 in
 \cite{BTV17}. Let $q=37,\kk=\ff_q,$ and let $X,Y,Z$ be copies of the projective line ${\mathbb P}^1$ with function fields $\kk(x),\kk(y),\kk(z),$ respectively.
 Suppose that $\phi_2: x\mapsto x^4,\, \phi_1: y\mapsto y^3,$ then
   $$
   \kk(x)\stackrel{\phi_2^\ast}\leftarrow \kk(y) \stackrel{\phi_1^\ast}\leftarrow \kk(z),
   $$
where $y^4-z=0$ and $x^{3}-y=0.$   
 Take $n=36,\nu=12,r_2=3,r_1=6,t=2,$ then $\dim(\cC)=12$ and $d_{\text{\rm min}}=18$ from Eq.~\eqref{eq:dred}, and thus the code $\cC$ has parameters
 [36,12,18] and is distance-optimal. The code is constructed as follows. 
Observe that 
   $$
   \phi_1^{-1}(1)=\{1,26,10\},\phi_1^{-1}(10)=\{7,34,33\},\phi_1^{-1}(26)=\{16,9,12\}.
   $$
These 9 points form the set of points on $Y$ used in the construction. Lifting them further to $X$, we obtain the fibers of the map $\psi^{-1}$ as follows:
  $$
  \begin{aligned}
    &D_1=(B_1^{(1)}=\{1,6,36,31\}) \cup (B_2^{(1)}=\{8,11,29,26\}) \cup (B_3^{(1)}=\{27,14,10,23\})\\
    &D_2=(B_1^{(2)}=\{2,12,35,25\}) \cup (B_2^{(2)}=\{16,22,21,15\}) \cup (B_3^{(2)}=\{17,28,20,9\})\\
    &D_3=(B_1^{(3)}=\{3,18,34,19\}) \cup (B_2^{(3)}=\{24,33,13,4\}) \cup (B_3^{(3)}=\{7,5,30,32\}).
    \end{aligned}
    $$
% Let
%    $$
%    \begin{aligned}
%    &B_1^{(1)}=\{1,6,36,31\}    & &B_1^{(2)}=\{2,12,35,25\} & &B_1^{(3)}=\{3,18,34,19\}\\
%    &B_2^{(1)}=\{8,11,29,26\}   &  &B_2^{(2)}=\{16,22,21,15\} &&B_2^{(3)}=\{24,33,13,4\}\\
%    &B_3^{(1)}=\{27,14,10,23\}  &  &B_3^{(2)}=\{17,28,20,9\}  &&B_3^{(3)}=\{7,5,30,32\}
%    \end{aligned}
%    $$
 The subset $D_\alpha, \alpha=1,2,3$ is the evaluation set of the middle code $\cC_\alpha$ with parameters
 $[12,6,6]$ and locality $3$. Finally, let us construct the set of functions \eqref{eq:V1}. Let $f=x^{12},y=x^4,$
 then the basis of functions is given by  
   $$\cF:=\{f^iy^jx^k=x^{12i+4j+k};\,i=0,1,j=0,1,k=0,1,2\}.$$ %Here the powers $1, x^{12}$ correspond to the functions $f_i$ in \eqref{eq:V}.
%   , while   the function $y$ is realized as $x^4.$ 
   
   The code is defined by the linear map $\kk^{12}\to\kk^{36}$ that sends a vector $(v_{ijk},i=0,1,2;j=0,1;k=0,1)$ to the set of
evaluations of the polynomial $v(x):=\sum_{i,j,k}v_{ijk}x^{12i+4j+k}$ at the points $a\in\ff_q^\ast.$
The code $\cC$ can correct up to 17 erasures by interpolating the polynomial $v(x)$ over $12$ points outside the erased set. %\textcolor{red}{(There exist 12 independent points. I do not see how to identify them easily.)}

At the same time, any 5 erasures can be corrected using a local repair procedure. In the worst case, these erasures are
located within a single fiber of $\psi:X\to Z$, say $D_1$. The polynomial $v(x)$ restricted to $D_1$ has at most $6$ nonzero
coefficients, and since the code $\cC_1=\cC|_{D_1}$ has distance 6, it can be interpolated from its values at $6$ (or fewer) points in $D_1$ 
outside the erased subset. Finally, any single erasure can be recovered from the 3 nonerased points in its repair group $B_i^{(\alpha)}$ because
the restriction of the code $\cC$ to $B_i^{(\alpha)}$ is a $[4,3,2]$ RS code obtained by evaluating a polynomial of degree $\le 2.$

To give an example, suppose that all $v_{ijk}=1,$ then $v(x)=(1+x+x^2)(1+x^4)(1+x^{12}).$ On the set $D_1$ this polynomial evaluates to
  $c_1:=(12, 24, 4, 13, 20, 4, 7, 0, 4, 17, 0, 30),$ where the order of locations is the same as in the set $|B_1^{(1)}|B_2^{(1)}|B_3^{(1)}|.$
Suppose that the value 20 is erased, which corresponds to location 8 in the set $B_2^{(1)}$. 
%The polynomial $(1+x^4)(1+x^{12})$ on the set $B_2^{(1)}$ is constant (equal to 17), and we can find the coefficients of the quadratic polynomial $a+bx+cx^2$ from its values 4,7,0.
The restriction of the polynomial $v(x)$ to the set $B_2^{(1)}$ is of degree 2, say $a_1+a_2 x+a_3 x^2,$ and we can find $a_1,a_2,a_3$ from the 
nonerased values in this set of positions. We obtain $(v(x))|_{B_2^{(1)}}=17(1+x+x^2)$ and recover the erased value by taking $x=8.$ 

Now suppose that the first 5 coordinates in the vector $c_1$ are erased. The restriction of $v(x)$ to the set $D_1$ is of the form
$a_1+a_2x+a_3x^2+(a_4 +a_5x+a_6x^2)x^4.$ Since $x^4$ is constant on $B_i^{(1)}$, finding the coefficients amounts to recovering two copies
of a quadratic polynomial. We can find them from the sets $B_i^{(1)}, i=2,3,$ each of which contains 3 independent evaluations, and we obtain
$(v(x))|_{D_1}=2(1+x+x^2)(1+x^4).$ Finally, we correct the erasures by
evaluating this polynomial at the erased locations. 

Note that taking the set of functions in the form $\cF_1:=\{x^{4j+i},j=0,1,2,3; i=0,1,2\},$ we would obtain a $[36,12,22]$ code with locality $3$ that belongs
to the code family of \cite{tamobarg}. By changing the functional basis from $\cF_1$ to $\cF$, we reduce the distance to 18 in exchange for adding the hierarchical locality
property. \hfill\qed}

%   \begin{gather*}
%   \phi_2(x)=\prod_{\xi\in B_1^{(1)}}(x-\xi)=x^4-1\\
%   \phi_2(x)=\prod_{\xi\in D_1}(x-\xi)=(x^4-1)(x^4+11)(x^4-10)=x^{12}-1.
%   \end{gather*}

  \vspace*{.1in}
 
 By  increasing the degree of the map $\phi_2$ we can increase the distance $\rho_2$ from 2 to larger values so that each
small repair group is resilient to more than one erasure. More specifically, let $\rho_2\ge 2$, and let $r_1,r_2$ be such
that $r_1=sr_2$ and $((s+1)(r_2+\rho_2-1))|n.$ Let $\phi_1,\phi_2\in \kk[x]$ be polynomials constant on their respective
repair groups, and let $\deg(\phi_2)=r_2+\rho_2-1$ and $\deg(\phi_1)=(r_2+\rho_2-1)(s+1).$ 
 Define the set of functions $V=\text{span}_{\kk}(\phi_1^k\phi_2^jx^i)$ where the indices vary as in \eqref{eq:V1}.
 Finally, construct the code $\cC$ as the set of evaluations of the functions in $V$ on the points in $D$.
The properties of $\cC$ are summarized in the following form.
 \begin{proposition}\label{prop:many}
 The code $\cC$ has length $n,$ dimension $tr_1$ and distance 
   $$
   d_{\text{\rm min}}(\cC)=n-tr_1+1-(t-1)(r_2+\rho_2-1)-(ts-1)(\rho_2-1).
   $$
It is an optimal H-LRC code with two levels of hierarchy and locality parameters $(r_1, r_2+2\rho_2-1),(r_2,\rho_2).$ 
The middle codes $\cC_\alpha$ are optimal $[(s+1)(r_2+\rho_2-1),r_1,r_2+2\rho_2-1]$ LRC codes.
 \end{proposition}
 
 It is also possible to increase the degree of the map $\phi_1$ thereby increasing the distance of the codes $\cC_\alpha$ while still keeping the distance $\rho_2=2.$ Finally, it is possible to increase the degrees of both the maps $\phi_1,\phi_2,$ thereby increasing both $\rho_1$ and $\rho_2$. As is easily checked, the resulting codes still retain the optimality properties.
 
% Doing so will also increase the distance $\rho_1$ to $r_2+2\rho_2-1$. In this way each of the codes $\cC_\alpha$ is an optimal LRC code, and the overall code $\cC$ has the locality parameters
% $(r_1,\rho_1),(r_2,\rho_2).$ As is easily checked, it retains the optimality property with respect to the bound \eqref{eq:sh}.

\section{H-LRC codes from automorphisms of curves}
While the previous section introduced a general construction of H-LRC codes on algebraic curves, so far we gave only one concrete
example that relies on maps between projective lines. To construct a class of examples, we develop the ideas put forward in a series of recent works in \cite{JinMaXing17,LiMaXing17b}, constructing towers of curves in the form of \eqref{eq:tow} from automorphism groups of curves.
Let $G$ be a subgroup of $\text{Aut}(X)$ with subgroup $H$ such that $|H|=r_2+1$ and $|G|=\nu.$
Let $\kk(X)^H$ be the set of $H$-invariant functions in $\kk(X)$ and let $\kk(X)^G$ be the same for $G$.
Consider the following tower of function fields:
  \begin{equation}\label{eq:ff}
  \kk(X)\xleftarrow{\phi_2^{\ast}}\kk(X)^H\xleftarrow{\phi_1^{\ast}}\kk(X)^G,
  \end{equation}
where $\phi_1^\ast,\phi_2^\ast$ are the embedding maps of the function fields. 
Let $g_1$ and $g_2$ be primitive elements of the extensions $\kk(X)^H/\kk(X)^G$ and $\kk(X)/\kk(X)^H,$ respectively. Choose places $Q=\{Q_1,\dots, Q_m\}$ of $\kk(X)^G$ that split completely in $\kk(X)$ (i.e., there are $\nu=(s+1)(r_2+1)$ places in $\kk(X)$ above each $Q_i$), and let $Q_{\infty}$ be a positive divisor with support disjoint from $Q$. Let $D$ be the collection of places in $\kk(X)$ above the places in $Q$.

Since \eqref{eq:ff} is a particular case of \eqref{eq:tow}, the general construction in \eqref{eq:hmap} applies. Using it, we obtain a code $\cC(D, \{\phi_1, \phi_2\})$ with parameters $[n,k,d]$ determined by Proposition \ref{prop:general}. Specifically, 
  \begin{equation}\label{eq:parameters}
  n=m\nu,\; k=r_2st, \;\;t:=\dim(L(Q_\infty))\ge 1,
  \end{equation}
the distance $d$ is bounded in \eqref{eq:d}, and the locality parameters equal $(sr_2,\rho_1),(r_2,2),$ where $\rho_1$ is given in
\eqref{eq:rho1}.
 
In what follows we give some specific examples.

\subsection{Automorphisms of rational function fields}
% In Sec.~\ref{sec:curves} we presented a general construction of
%H-LRC codes obtained from covering maps of projective lines. This construction relies on certain divisibility assumptions, and we made no effort to establish when those are satisfied. Here we use the construction relying on \eqref{eq:ff} to give specific
%examples of the codes in Prop.~\ref{prop:RS} that relies on the the structure of the automorphism group of the rational function field.
Let $\kk(X)=\kk(x)$ be a rational function field. Let us assume that $r_2$ and $s$ are such that there exists 
a subgroup $G$ of $\text{Aut}(X)=\text{PGL}_2(q)$ of order $(r_2+1)(s+1).$
We apply the construction \eqref{eq:ff} above to get a tower of rational curves 
\begin{equation*}
    X\xrightarrow{\phi_2}Y\xrightarrow{\phi_1}Z.
\end{equation*}
By construction, both the degrees of $x$ and $y$ are 1. We obtain
an H-LRC code $\cC$ with parameters $((r_2s, \rho_1),(r_2,2))$ where on account of \eqref{eq:rho1},
  $$
  \rho_1 \ge \nu-(s-1)(r_2+1)-(r_2-1)=r_2+3,
  $$ 
Note that this is in fact an exact equality because of the upper bound \eqref{eq:sb}.   
  Moreover, as is easily checked, the
  code $\cC$ as a whole meets the upper bound \eqref{eq:sh} with equality. We obtain:
  \begin{proposition} \label{prop:P1}
  Let $n\le q$ be a multiple of $(r_2+1)(s+1).$
Using construction \eqref{eq:ff} for the subgroups of the automorphism group of the rational function field, we obtain optimal 
$[n,k,d]$ H-LRC codes with parameters $((sr_2, r_2+3),(r_2,2)).$ 
\end{proposition}

These codes are in fact from the same family as the codes constructed in Prop.~\ref{prop:RS}. However, we can extend this construction
to optimal H-LRC codes of length $q+1$ relying in part on the ideas in \cite{JinMaXing17}. Assume that $G<\text{PGL}_2(q), |G|=\nu|(q+1),$
then there exists a subset $S$ of $m:=(q+1)/\nu$ rational places of $\kk(X)^{G}$ that split completely in $\kk(X).$ Let $S=(Q_1,\dots,Q_m)\subset Z(\kk)$ and let $H<G,|H|=r_2+1.$ Let $P_\infty$ be the infinite place in $\kk(X).$ W.l.o.g. we can
assume that $P_\infty|Q_1.$ Let $(y)_\infty$ be the polar divisor of $y$ and assume that $\text{supp}((y)_\infty)\cap \phi_2^{-1}(S)=\emptyset.$ As above, let the set of evaluation points be $D=\cup_{i=1}^m\psi^{-1}(Q_i)$, and let 
the fiber above $Q_1$ be $P_{11}=P_\infty,P_{12},\dots,P_{1,\nu}.$ 
The code $\cC$ is constructed by evaluating the functions in
\eqref{eq:V} at the points in $D$. Specifically, $\cC$ is the image of the following map:
  $$
  f\in V \mapsto ((x^{-r_2+1} f)(P_{11}),f(P_{12}),\dots,f(P_{m\nu}))\in \kk^{q+1}.
  $$  
The idea of constructing codes on curves whose set of evaluation points $D$ includes the support of $Q_\infty$ (by multiplying
by an appropriate degree of the uniformizing parameters) has appeared in the literature, e.g., \cite[p.194]{TVN07}.  
\begin{proposition} The locality parameters of the code $\cC$ are $(sr_2,r_2+3),(r_2,2)$, making it into an optimal
2-level $q$-ary H-LRC code of length $q+1.$
\end{proposition}
\begin{proof} We only need to check that the small (size-$(r_2+1)$) recovering set that contains $P_{11}$ supports local 
correction. If the erased coordinate is $P_{11},$ then its value can be found by regular polynomial interpolation. Otherwise,
observe that the function $f$ on this set has the form 
   $f(x)=\sum_{k=0}^{r_2-1}a_kx^k,$ where $a_k$'s are constants. Observe that $a_{r_2-1}=(x^{-r_2+1} f)(P_{11}).$ The remaining
   $r_2-1$ coefficients of $f$ can be found by Lagrange interpolation from the other $r_2-1$ evaluations of $f$ in this set.
\end{proof}
For instance, one can take $n=q+1=28,$ obtaining an optimal $[28,6t,d=37-14t]$ H-LRC code over $\ff_{3^3}$ with locality parameters 
$(r_1=6,\rho_1=9),(r_2=6,\rho_2=2).$ Nontrivial examples arise when $t=1,2,$ and we obtain codes with the parameters $[28,6,23],[28,12,9]$ that meet the bound \eqref{eq:sh}.

\section{H-LRC Codes of Length $n> q+1$ constructed from elliptic curves} \label{sect:elliptic}
\subsection{LRC codes from quotients of elliptic curves} 
Li et al.~\cite{LiMaXing17b} introduced a construction of optimal LRC codes on elliptic curves obtained from quotients of the elliptic curve by 
subgroups of automorphisms. We present this construction in this section and extend to H-LRC codes in the next one.
% The construction is generally based on LRC construction from maps between curves introduced in \cite{tamobarg}. 

Let $E$ be an elliptic curve over $\kk=\mathbb{F}_q$ and let $G$ be a subgroup of the automorphism group $\text{Aut}(E)$. 
%We restrict ourselves to the case $\text{char}(\kk)=2$
Note that the automorphism group is the largest for $\text{char}(\ff_q)=2,3$, and therefore examples given in \cite{LiMaXing17b}
are given for these cases. In the H-LRC case, since two levels of hierarchy are required, most useful examples arise 
in the characteristic 2 case when the automorphism group is of size 24 \cite{Silverman09}.

Let us assume that 
$|G|=r+1=2s.$ 
%and that $G$ contains the negation map $\sigma:(x,y)\to(x,-y)$. 
Denote the coordinate 
functions of the automorphisms in $G$ by $\sigma_i((x,y))=(f_i(x,y), g_i(x,y))$. Assume that the set of 
$x$-coordinate functions $f_i$ has size $s$ (in the case of odd characteristic this can be achieved by including
in $G$ the negation map on $y$, i.e., the automorphism $\sigma:(x,y)\to(x,-y)$). Let us index the automorphisms $G=\{\sigma_1, \cdots , \sigma_{r+1}\}$ so that
to ensure that $f_{i+s}(x,y) = f_{i}(x,y)$. Finally, let us 
assume that there is a point $P=(a,b)$ on $E$ such that the points $P_i = \sigma_i(P)$ are distinct, i.e., $P$ is contained in a 
totally split fiber of the covering map $\phi: E \to E/G$.

%Our next goal will be to define the set of functions for our evaluation code \eqref{eq:V},\eqref{eq:hmap}. 
Let us define a function 
  $$ 
  z(x,y)=\displaystyle \prod_{i=1}^s \frac{1}{f_i(x,y)-a}. 
  $$
First we note that $\sigma_i(z) = z$ for all $1\le i \le r+1$. This means that $z$ can be thought of as a function in $k(E/G)$. More importantly, this implies that $z$ is constant on fibers of the covering map $\phi$. 
Powers of the function $z$ will take the place of the functions in the Riemann-Roch space $L(D)$ in the 
general construction of Sec.~\ref{sec:curves}.
 Also note that the divisor of $z$ is
   $$
   (z) = (r+1)\infty - P_1 \cdots - P_{r+1}.
   $$ 
Define functions $w_0=1$ and $w_i, i=1,\dots,r-1$ in 
$L_i =  L(P_1+\cdots+P_{i+1})$ such that $L_i = \text{span}\{1, w_1,\cdots, w_i\}$ for $1 \le i \le r-1$. Such a choice
is always possible by the Riemann-Roch theorem. Define the space of functions used to construct the code as follows:
$$ 
V = \text{span}\{ (z^{t}, w_iz^j)\, |\,i=1,\dots,r-1, 0\le j\le t-1 \}.
$$
%The reason that we isolate the function $z^t$ is very important. 
%Not multiplying $z^t$ by any $w_i$'s guarantees that $V\subseteq L(t(P_1+\cdots+P_{r+1}))$, and it is this condition that
%allows the codes constructed in\cite to be optimal. 
Let $Q = \{Q_1,\ldots, Q_n\}$ be a union of totally split fibers of the covering map $\phi$ that does include the fiber formed
by the points $P_i.$
The LRC code is obtained from the evaluation map
\begin{align*}
\text{ev}:&V\to \kk^n\\
&f \mapsto (f(Q_1), \ldots, f(Q_n))
\end{align*}
As shown in \cite{LiMaXing17b}, the resulting codes are optimal with respect to \eqref{eq:sb2}. The recovering sets of the code are coordinates contained in the same fiber of $\phi$. Restricted to a fiber of $\phi,$ a function in $V$ becomes just a linear combination of the $r$ linearly independent functions $w_i$, enabling one to recover the missing coordinate.

%The parameters of the code are mostly straightforward. The codes is an LRC $(n,k,d)$ codes with $k=rt+1$ and $d\ge n- t(r+1)$. If we apply the singleton bound for LRC codes we also get $d\le n-k - \lceil \frac{k}{r} \rceil + 2 = n - (tr+1) - \lceil \frac{tr+1}{r} \rceil +2 = n-tr-t = n-t(r+1)$. This implies that the code construction above gives us optimal LRC codes.

{\em Remark:} Even though \cite{LiMaXing17b} did not go beyond the genus 1 case, the above 
construction can be extended to curves of genus 2 with a only 
few changes to the definition of the $w_i$'s. Namely, take $w_0=1$ as before and let $w_i$ to 
be a nontrivial function in $L_i=L(P_1+\cdots+P_{i+2})$ such that $L_i=\text{span}\{w_0,\ldots,w_i\}$ The advantage in applying this construction to genus 2 curves is that they can have larger automorphism groups and more rational points, allowing greater flexibility in choices of the parameters. In particular, paper \cite{LiMaXing17b} gives examples of the above construction for maximal elliptic curves that result in optimal LRC codes of length close to $q+2\sqrt{q}$. With genus two curves we can easily construct 
optimal LRC codes of length $n$ close to $q+4\sqrt{q},$ which constitute a family of optimal LRC 
codes of length larger than reported in the literature (apart from the case of $d=3,4$ in \cite{LuoXingYuan18}). At the same time,
so far we have not been able to extend this observation to the case of H-LRC codes.

\subsection{H-LRC Codes from quotients of elliptic curves} Let $E$, $G$, $\{P_i\}$, $\{Q_i\}$ and $z$ be as above. Additionally choose a subgroup $H \le G$ of order $r_2+1$. Let $\bar{P_i}$ be the point on $E/H$ below $P_i$. Let $m+1:=(r+1)/(r_2+1)$ and suppose the $P_i$ are enumerated such that $\bar{P}_1, \ldots , \bar P_{m+1} $ are all distinct.

If $E/H$ is of genus 1, we take $w_0 = 1$ and $w_i$ to be a function in $\bar L_i = L(\bar P_1+\cdots + \bar P_{i+1})$ for $1\le i \le m-1$ such that $\bar{L_i} = \text{span}\{1, w_1, \ldots, w_i\}$ as before. Otherwise, if the genus of $E/H$ is 0, we take
 $w_0 = 1$ and $w_i$ to be a function in $\bar L_i =L(\bar P_1+\cdots + \bar P_{i+1})$ for $1\le i \le m-1$ such that $\bar L_i  = \text{span}\{1, w_1, \ldots, w_i , w'_i\},$ where $w'_i$ is any additional linearly independent function in the Riemann-Roch space
 $\bar L_i.$ Note that none of the $w_i$'s have poles at $\bar P_{m+1} $.

Let $P_{m+1,1}, \ldots, P_{m+1, r_2+1}$ be the points on $E$ above $\bar P_{m+1}$. Take $y_0=1$ and $y_i$ to 
be a function in $L_i=L(P_{m+1, 1}+ \cdots + P_{m+1, i+1})$ such that $L_i = \text{span}\{1,y_1,\ldots,y_i\}$. For clarity, we will define the space of functions in two steps. Define $V'$ and $V$ as follows:
   \begin{gather*}
V' = \text{span}\{w_{m-1}, w_jy_k | 0 \le j \le m-2, 0\le k \le r_2-1\} \\
V  = \text{span}\{z^t, z^ig | 0\le i \le t-1, g\in V'\}
\end{gather*}
Once again the code $\cC$ is obtained by evaluating the points in $Q$ at all the functions in $V$. 
Construct the code $\cC$ evaluating the functions in $V$ at the points in $Q$ (cf.~\eqref{eq:hmap}).
   \begin{proposition}
The code $\cC$ constructed above is an $[n,k,d]$ H-LRC code with locality parameters $ ((r_1, \rho_1),(r_2,\rho_2=2)) $ where
    \begin{align*}
    &r_1 = r_2(m-1)+1\\
    &r_2+1\leq \rho_1 \leq 2r_2+2\\
    &k=t(r_2(m-1)+1)+1\\
    &d \geq n-(t(m+1)(r_2+1)-(r_2+1)).
\end{align*}
\end{proposition}
\begin{proof}
The middle codes have length $\nu = (m+1)(r_2+1)$ and dimension $r_1 = \dim(V') = r_2(m-1) + 1$ since the function $z$ is constant on the fibers of $E\to E/H$. Also, restricted to a fiber, the functions in $V'$ 
are contained in $L(\bar P_1+ \cdots + \bar P_{m}) \cup L(\bar{P}_1+\cdots + \bar{P}_{m-1} + \bar{P}_{m+1}) $. This implies
 that the minimum distance of the middle codes satisfies 
 $\rho_1\ge \nu - m(r_2+1) = r_2+1$. The upper bound on $\rho_1$ follows from the Singleton bound \eqref{eq:sb2}.
 
The value of the dimension $k$ follows directly from the construction. Finally, since 
$V\subseteq L(t(P_1 + \cdots + P_{r+1}) - \bar{P}_{m+1})\cup L(t(P_1 + \cdots + P_{r+1}) - \bar{P}_{m})$ we have 
$$d>n - (t(m+1)(r_2+1) - (r_2+1)).$$
\end{proof}

\vspace*{-.2in}

\subsection{Examples:} {For any even $m$ there exists $\gamma\in \ff_{2^m}$ such that the elliptic curve $E: y^2 + y = x^3 + \gamma$ 
is maximal} in the sense that the number of rational points on $E$ meets the Hasse-Weil bound \cite[Lemma 3.3]{LiMaXing17b}. The automorphism group of $E$ is of order 24, which is also maximal since an elliptic curve can have at most 24 automorphisms. The automorphisms are given by the following coordinate functions:
   $$
    \sigma_x(x,y)=u^2x+s,\\\quad
    \sigma_y(x,y)=y+u^2sx + t,
   $$
where $u^3=1, s^4+s=0, t^2+t+s^6=0$. The subgroup $G$ of Aut$(E)$ given by restricting $s$ to be 0 or 1 is order 12 and we take $H$ to be the order 4 subgroup of $G$ given by further restricting $u$ to be 1. By the Riemann-Hurwitz \cite[p.37]{Silverman09} formula we have
$$
    2g(E)-2 \ge 2g(E/G)-2 + \sum_{P\in E(K)}(e_P-1),
$$
where $g(E)$ and $g(E/G)$ are the genus of $E$ and of $E/G,$ respectively, and $e_P$ is the ramification index of the point $P$. 
Note that we use the Riemann-Hurwitz formula in the inequality form because in characteristic 2 some of the points are
wildly ramified. For instance, let us take $q=64.$
Since the point at infinity is totally ramified, the above equation implies that in the worst case
there are 13 additional ramified affine points on $E$ and therefore, there are at least 67 unramified points.
Since the order of $G$ is 12, this implies that there are in fact at least 72 unramified points. This results in 
at least $60$ evaluation points on $E$. The general code construction in this case gives an $[n=60, k=4t+1, d] $ H-LRC code with 
locality parameters $((4,\rho_1),(3,2))$ where $4\le\rho_1\le7$ and
  \begin{gather*}
    d \ge n-12t+4,     1\le t\le5.
\end{gather*}
Note that we do not have enough information to determine the distance of the ``middle'' codes $C_1$, making it difficult to compare
the value of $d$ with the upper bound \eqref{eq:sh}. Substituting $\rho_1=4$, we obtain 

\vspace*{.1in}
\begin{center}\begin{tabular}{ccc}
$t$ &$k$ &$d$\\
1 &5&$52\le d\le 53$\\
2 &9&$40\le d\le 46$\\
3 &13&$28\le d\le 38$.
\end{tabular}
\end{center}\vspace*{.1in}
 To obtain examples of length $n> q$, we should take a larger-size field, for instance let us take $\ff_{256}$. Applying the same arguments as above, we obtain H-LRC codes with 
 parameters $[264, 4t+1, d]$ and locality $((4,\rho_1),(3,2))$ where $4\le\rho_1\le 7$ and
 \begin{gather*}
    d \ge n-12t+4, 1\le t\le 22.
\end{gather*}

\section{Some families of curves and associated H-LRC codes}\label{sect:examples}
While Proposition \ref{prop:general} gives a general approach to constructing H-LRC codes, estimating the parameters 
for a given curve is a difficult question, in particular because controlling the multiplicity $\deg_{\psi}(x)$ in \eqref{eq:rho1} is
not immediate. The largest distance $\rho_1$ is obtained if the function $x$ is injective on the fibers of $\psi$, i.e., if $\deg_{\psi}(x)=1.$ In this section we present two general constructions that make this possible using properties of the automorphism groups of 
curves. Thus, all the H-LRC code families constructed below in this section share the property of having distance-optimal middle codes.

\subsection{Kummer curves} The simplest and at the same time rather broad class of examples arises when $G < \text{Aut}(X)$ is a cyclic group of order not divisible by the characteristic, i.e., when $X$ is a Kummer curve.

Recall that a Kummer curve $X$ over $\kk=\ff_q$ is defined by the equation
   \begin{equation}\label{eq:Kummer}
   y^m=f(x),
   \end{equation}
   where $m|(q-1)$ and $f(x)\in K:=\ff_q(x)$  \cite[pp.122ff.]{Stichtenoth09}, \cite[p.168]{TVN07}. The field $L:=\ff_q(x,y)$ is a degree $m$ cyclic extension of $K$, and any cyclic extension
   of degree $m$ can be written in this form. The following examples of Kummer curves are maximal and lead to H-LRC codes
   with good parameters.
\begin{enumerate}
\item The {\em Hermitian curve} $X:y^{{q_0}+1}=x^{q_0}+x$ over the field $\ff_q,q=q_0^2$ is a maximal Kummer curve.
\item The {\em Giulietti-Korchm{\'a}ros curves} \cite{GK09} are given by the affine equation
   $$
   y^{q_0^3+1}=x^{q_0^3}+x-(x^{q_0}+x)^{q_0^2-q_0+1},
   $$
and have genus $g=\frac12(q_0^3+1)(q_0^2-2)+1.$ They are maximal over $\ff_q$ for $q=q_0^6$. 
\item (The {\em Moisio curves} \cite{Moi04})  Let $h\in \{0,\dots,l\}$, let $m|(q_0^l+1)$ and let $q=q_0^n.$
Let $L$ be an $\ff_{q_0}$-subspace of dimension $h$ in $\ff_q$ and suppose that
  $$
  \prod_{\alpha\in L}(x-\alpha)=\sum_{i=0}^h a_ix^{q_0^i}.
  $$
Let 
    $$
      R(x)=\sum_{i=0}^h a_i^{q_0^{2n-i}}x^{q_0^{h-i}}.
    $$
Then the curve given by $y^m=R(x)$ is maximal over $\ff_{{q^2}}$ of genus $(m-1)(q_0^h-1)/2,$ so 
  $$
  |X(\ff_q)|=q_0^{2l}+(m-1)(q_0^{l+h}-q_0^l)+1.
  $$
\end{enumerate}

%\begin{example} {\rm The {\em Hermitian curve} $X:y^{{q_0}+1}=x^{q_0}+x$ over the field $\ff_q,q=q_0^2$ is a maximal Kummer curve.}
%%(the H-LRC codes on $X$ are analyzed in Example \eqref{eq:H} using a slightly different approach).
%\end{example}
%
%\begin{example} {\rm The {\em Giulietti-Korchm{\'a}ros curves} given by the affine equation
%   $$
%   y^{q_0^3+1}=x^{q_0^3}+x-(x^{q_0}+x)^{q_0^2-q_0+1}
%   $$
%are maximal over $\ff_q$ for $q=q_0^6$ \cite{GK09}. Given $q_0,$ the genus equals $g=\frac12(q_0^3+1)(q_0^2-2)+1.$}
%\end{example}
%
%\begin{example} {\rm (The {\em Moisio curves} \cite{Moi04})  Let $h\in \{0,\dots,n\}$, let $m|(q_0^n+1)$ and let $q=q_0^n.$
%Let $L$ be an $\ff_{q_0}$-subspace of dimension $h$ in $\ff_q$ and suppose that
%  $$
%  \prod_{\alpha\in L}(x-\alpha)=\sum_{i=0}^h a_ix^{q_0^i}.
%  $$
%Let 
%    $$
%      R(x)=\sum_{i=0}^h a_i^{q_0^{2n-i}}x^{q_0^{h-i}}.
%    $$
%Then the curve given by $y^m=R(x)$ is maximal over $\ff_q$.}
%\end{example}

{Let $G_0=\text{Gal}(L/K)$ be the cyclic group of order $m$. The action of $G_0$ on the curve $X$ is given by $(x,y)\mapsto(x,\alpha y),$ where $\alpha\in \ff_q,\alpha^m=1.$ If $m$ is well-decomposable, say $m=(a+1)(b+1)c,$ then one can easily
find subgroups
   $H<G<G_0\subseteq \text{Aut}(X)$
with desirable properties. Indeed, let $m$ be as above and let $\alpha$ be a generator of $G_0$. Then we can take $G=\langle\alpha^c
\rangle, H=\langle\alpha^{(b+1)c}\rangle, |G|=(a+1)(b+1), |H|=a+1.$
It is clear that the invariants of any subgroup of $G_0$ are generated by powers of $y,$ for instance, from \eqref{eq:Kummer}, $y^{a+1}$ is fixed by
any power of $\alpha^{(b+1)c}$, etc. }

{Specializing the construction \eqref{eq:ff}, we obtain
  $$
  \kk(X)=\kk(x,y)\hookleftarrow \kk(X)^H=\kk(x,y^{a+1})\hookleftarrow \kk(X)^G=\kk(x,y^{(a+1)(b+1)})
  $$
Now it is clear that the primitive element $y$ is injective on the fibers of $\phi:X\to X/G$, and we can use the general code construction with $\deg_\psi(y)=1.$ }

{Using the general construction of Proposition \ref{prop:general} for the curves listed above, we obtain 
several families of H-LRC codes. The case of Hermitian curves is analyzed below in Section \ref{sect:GS} in the context
of power maps (see
Example \ref{ex:H}). 

Turning to the Giulietti-Korchm{\'a}ros curves, we observe that the total number of rational points on the curve $|X(\ff_q)|$ equals $q_0^8-q_0^6+q_0^5+1$ (which meets the Hasse-Weil
bound $N(X)\le q+1+2\sqrt q g$). {Setting aside the point at infinity, we observe that the projection map on $x$ is ramified in at
most $q_0^3$ places, leaving $n\ge q_0^8-q_0^6+q_0^5-q_0^3$ totally split places which form the evaluation set $D$.} Now we use 
Proposition \ref{prop:general} to claim the existence of H-LRC codes with the following parameters:
\begin{gather*}
n\ge (q_0^5-q_0^3)(q_0^3+1) %q_0^8-q_0^6+q_0^5
, \; k=\dim(L(Q_\infty))ab\\
d\ge n-\deg(Q_{\infty})(a+1)(b+1)-q_0^3(ab+b-2)\\
  r_2 = a,\hspace{5pt} \rho_2 = 2\\
    r_1 = ab,\hspace{5pt} \rho_1 = a+3.
\end{gather*}
(note that $\deg(y)=q_0^3$). 
To obtain specific examples, we may take $q_0=4,$ getting $a=4,b=12,c=1$ or $q_0=17,$ in which
case the decomposition $q_0^3+1=2\cdot27\cdot7\cdot 13$ leaves multiple options for the localities of the codes, etc.}
We note that the distance of the middle codes is the largest possible, meeting the bound \eqref{eq:sb} with equality.

For the Moisio curves, the size of the ramification set is at most $q_0^h$, leaving at least $q_0^{2l}+(m-1)(q_0^{l+h}-q_0^l)-q_0^h$
points for the evaluation set $D$.
The codes from the Moisio curves {are constructed over $\ff_{q^2}$} and have the following parameters:
\begin{gather*}
n\ge q_0^{2l}+(m-1)(q_0^{l+h}-q_0^l)-q_0^h,\; k=\dim(L(Q_\infty))ab\\
d\ge n-\deg(Q_{\infty})(a+1)(b+1)-q_0^h(ab+b-2)\\
  r_2 = a,\hspace{5pt} \rho_2 = 2\\
    r_1 = ab,\hspace{5pt} \rho_1 = a+3.
\end{gather*}
For instance, we can take $q_0=2,l=5,$ and then taking $m=q_0^l+1$, we obtain H-LRC codes with localities $r_1=10,r_2=20,$ etc.

\subsection{Artin-Schreier curves} Let $q=q_0^e$ for some $e\in{\mathbb N}.$ A curve with the affine equation 
   \begin{equation}\label{eq:as}
   y^{q_0}-y=f(x)
   \end{equation}
for $f(x)=\ff_q(x)$ is called an {\em Artin-Schreier curve} \cite[pp.127ff.]{Stichtenoth09}, \cite[p.173]{TVN07}. More generally, a {\em generalized} Artin-Schreier
curve is given by the equation
   $$
   P(y)=f(x),
   $$
   where $P(y)=a_uy^{q_0^u}+a_{u-1}y^{q_0^{u-1}}+\dots +a_0 y, a_0\ne 0$ is a linearized polynomial whose roots form a linear subspace of $\ff_q$. Such a curve $X$ forms a Galois covering of the projective line with the Galois group $G_0:=\text{Gal}(X/{\mathbb P}^1)
   \cong {\cL}(P)$ where $\cL(P)$ is a linear space of roots of $P(y)$ in $\ff_{q_0}$ (thus, for coverings of the form \eqref{eq:as}, $G_0\cong \ff_{q_0^u}^+$). The group $G_0$ acts on the points of $X$ by $(x,y)\mapsto(x,y+\alpha)$ for $\alpha\in G_0.$
Artin-Schreier covers give many examples of curves that are either maximal or close to maximal. Examples of maximal curves include
the following families.
\begin{enumerate}
\item The Hermitian curves given by the equation $y^q+y=x^{q+1}$ over $\ff_{q^2},$
\item The Moisio curves (to see that they are Artin-Schreier, interchange $x$ and $y$ in their definition).
\end{enumerate}

 These examples are maximal in the sense that they attain the Hasse-Weil bound on the number of points.

\begin{enumerate}
\item[(3)] The Suzuki curves given by 
   $$
   S_q: y^{q}+y=x^{q_0}(x^q+x)
   $$
where $q_0=2^n,q=2^{2n+1}$ \cite{Hansen90}. The genus $g(S_q)=q_0(q-1)$ and the number of $\ff_q$-points is $N(S_q):=|X/\ff_{q}|=q^2+1$
(i.e., they fill the entire affine plane over $\ff_q$). The Suzuki curves are maximal because $N(S_q)$ meets the Oesterl{\'e} bound
for their genus. 
The full group $\text{Aut}(S_q)$ is the Suzuki group (hence the name), and it contains a subgroup isomorphic to $\ff_q^+$ which 
acts as before by $y\mapsto y+\alpha.$
\end{enumerate}

In each of the cases (1)-(3) above we have
   $$
   \text{Aut}(X)\supseteq G\cong (\integers/p\integers)^{e_2}\supset H\cong(\integers/p\integers )^{e_1}
   $$
for $q=p^e\ge 9$ and some exponents $e,e_1,e_2.$

Determining the primitive elements of the extensions in \eqref{eq:ff} with the above choice of $G$ and $H$ is generally not an easy question. We limit ourselves to two
simple examples.

\begin{enumerate}
\item Let 
  \begin{equation}\label{eq:X}
  X: \;y^q-y=f(x)
  \end{equation}
where $q=r^2,r=p^m\ge 3,$ and let $G\cong(\integers/p\integers)^{2m}, H\cong (\integers/p\integers)^m.$ In this case $G$ acts on $\kk(x,y)$ by fixing $\kk(x)$, i.e., we have $Z={\mathbb P}^1$ in \eqref{eq:tow} or $\kk(x,y)^G=\kk(x)$ in \eqref{eq:ff}. Let $H$ be a copy of $(\integers/p\integers)^m$ in $\ff_q^+$ with the property that $\alpha^r=-\alpha$ for all $\alpha\in H.$ 
In other words, $G\cong \ff_r^+\oplus\alpha \ff_r^+, H\cong \alpha\ff_r^+,$ where $\alpha^r=-\alpha.$ Further, let $z=y^r+y,$ then $z$ is invariant under the action $y\mapsto y+\alpha:$
   $$
   (y+\alpha)^r+(y+\alpha)=y^r+y=z.
   $$
 Further, 
  $$
  z^r-z=(y^r+y)^r-(y^r+y)=y^q-y=f(x),
  $$  
and thus, $\kk(x,y)^H=\kk(x,z),$ and \eqref{eq:ff} takes the form $\kk(x,y)\supset \kk(x,z) \supset \kk(x).$

On account of \eqref{eq:parameters}, we obtain a family of 2-level $[n,k,d]$ H-LRC codes, where $n=m\nu, k=r_2st,$ and $\nu=r^2,r_2=s=r-1,r_1=(r-1)^2,\rho_1=r+2,\rho_2=2.$

%\begin{gather}
%k(y,x)\supset k(z,x)=k(y,x)^H\supset k(x)=k(y,x)^G,\\
%G=F_q=F_r\oplus \alpha F_r, H= \alpha F_r,\\
%y^q-y=f(x)=z^r-z, z=y^r+y;\\
%k(y,x)\supset k(z,x)=k(y,x)^H\supset k(x)=k(y,x)^G,\\
%G=F_q=F_r\oplus \alpha F_r, H= F_r,\\
%y^q-y=f(x)=z^r+z, z=y^r-y;
%\end{gather}

\vspace*{.1in}\item Let us again take $X$ in the form \eqref{eq:X} where this time $q=r^3,r=p^m\ge 3,$ and let $G\cong (\integers/p\integers)^{3m},H\cong (\integers/p\integers)^{m}.$ Let $z=y^r-y$ and note that $z$ is fixed by the action of $H$ on $\kk(x,y),$
and thus $\kk(x,y)^H=k(x,z).$ Further,  
   \begin{align*}
   z^{r^2}+z^r+z=y^q-y=f(x).
   \end{align*}
The tower \eqref{eq:ff} has the form $\kk(x,y)\supset \kk(x,z) \supset \kk(x)$ since $G$ fixes the rational function field in $\kk(x,y).$

On account of \eqref{eq:parameters}, we obtain a family of 2-level $[n,k,d]$ H-LRC codes, where $n=m\nu, k=r_2st,$
and $\nu=r^3,r_2=r-1,s=r^2-1,r_1=sr_2,\rho_1=r+1,\rho_2=2.$

This example can be further generalized to the curve $X$ of the form \eqref{eq:X}, where $q=r^h, r=p^m$ and
$G\cong \ff_q^+, H\cong\ff_r^+.$
The tower \eqref{eq:ff} that gives rise to the code family, has the form $\kk(x,y)\supset \kk(x,z) \supset \kk(x),$ where $z=y^r-y$ and
   $$
   z^{r^h}+z^{r^{h-1}}+\dots+z=y^q-y=f(x)
   $$
We obtain a family of 2-level H-LRC codes with the parameters $\nu=r^h,s=r^{h-1}-1,r_2=r-1,r_1=sr_2,\rho_1=r+2, \rho_2=2.$
\end{enumerate}

\begin{remark} One can consider ``mixed'' Artin-Schreier--Kummer curves of the form $P(y^m)=f(x)$ over $\ff_q$ where $P$ is a
linearized polynomial and $m|(q-1),$ and apply arguments similar to the above. However, we are not aware of good examples of
such curves although is it likely that they exist.
\end{remark}
\begin{remark} It is also clear that the above construction can be generalized to more than two levels of hierarchy. Accomplishing this
depends on the factorization of $q-1$ for the Kummer case and does not require new algebraic ideas. A similar observation applies
to the Artin-Schreier case.\end{remark}

\section{H-LRC codes from the Garcia-Stichtenoth tower} \label{sect:GS}
In this section we use the general construction of H-LRC codes for curves in the GS tower. We begin by directly applying
the idea of Section \ref{sec:curves} and consider mappings
between the curves two levels apart in the tower, viz. \eqref{eq:tow}. This approach meets a complication in that it is not easy
to find the multiplicity $\deg_\psi(x).$ We circumvent this difficulty using power maps in Section~\ref{sec:power}, which are 
related to the constructions from Kummer covers in the previous section.

\subsection{Naive construction}
Let $q=q_0^2$ be a square and $\kk=\mathbb{F}_q$. For any $l \geq 2$ define the curve $X_l$ inductively as follows:
%\begin{align}\label{eq:gsdef}
%x_0 &:= 1, X_1 = \mathbb{P}^1, \kk(X_1)=\kk(x_1);\\ 
%X_l: z_l^{q_0} + z_l &= x_{l-1}^{q_0+1}, \text{ where } \text{ for } l \geq 3 \\
%x_{l-1} &:= \frac{z_{l-1}}{x_{l-2}}.
%\end{align}
 \begin{equation}
\left.\begin{array}{c}
 x_0 := 1, X_1 = \mathbb{P}^1, \kk(X_1)=\kk(x_1);  \\[.05in] X_l: z_l^{q_0} + z_l = x_{l-1}^{q_0+1}, \text{ where } \text{ for } l \geq 3 \\ [.05in]
 x_{l-1} := \frac{z_{l-1}}{x_{l-2}}.
 \end{array}\right\} \label{eq:gsdef}
 \end{equation}
%Note that the function field $\kk(X_l)=\kk(x_1,\dots,x_l),$ where $x_i,i=2,\dots,l$ is the primitive element of the
%algebraic extension $\kk(X_i)/\kk(X_{i-1}).$ 
The curves $X_l,l\ge 2$ form a tower of asymptotically maximal curves \cite{garcia95}.

The authors of \cite{BTV17} constructed LRC codes from covering maps between consecutive curves in this tower.
%, and \cite{LiMaXing17a} further developed this idea. 
Similarly, we will construct H-LRC codes with 2-fold hierarchy by extracting sub-towers of 3 curves from the full tower. 
Let $\phi_l:X_l\to X_{l-1}$ be the natural projection on the coordinates $x_i,i=1,\dots,l-1.$
Consider the following subtower of curves with their projection maps:
\begin{equation}\label{eq:gsmap}
X_{j+2} \xrightarrow{\phi_{j+2}} X_{j+1} \xrightarrow{\phi_{j+1}} X_j.
\end{equation}
Let $x=x_{j+2}$ and $y=x_{j+1}$ be primitive elements such that $\kk(X_{j+2})=\kk(X_{j+1})(x)$ and $\kk(X_{j+1})=\kk(X_{j})(y)$
 (see \eqref{eq:tow}). In this case
$\deg(y)=q_0^{j}$ and $\deg(x)=q_0^{j+1}$ are the degrees of the maps $X_{j+i}\to {\mathbb P}^1,i=1,2,$ respectively.
Let $S$ be formed of all the affine points of $X_j(\kk)$ that map to $\kk^\ast$ under the map $\phi_1 \circ\dots \circ \phi_j.$
Let $n_j=q_0^{j-1}(q_0^2-1)$ be the size of $S$, i.e., number of points above $\kk^\ast$ on $X_j, j=1,2,\dots.$ 
Let $Q_{\infty,j}$ to be the point at infinity on $X_j$ and let $t=\dim(L(\ell Q_{\infty,j})),$ where $g_{j}\le \ell\le n_{j}$
and $g_{j}$ is the genus of $X_{j}.$ Finally, denote $\psi_{j+2}=\phi_{j+1}\circ \phi_{j+2}.$

Using the general construction of Sec.~\ref{sec:curves} for the tower of curves  described in \eqref{eq:gsmap}, it is possible
to obtain a family of linear H-LRC codes with two levels of hierarchy. 
\begin{proposition} \label{prop:GS} For any $j\ge 1$
there exists a family of H-LRC codes with the parameters $[n,k,d]$ and locality $(r_1,\rho_1),(r_2,\rho_2=2),$ where
\begin{gather*}
n=q_0^{j+1}(q_0^2-1)\\
k =t(q_0-1)^2\ge (\ell-g_{j}+1)(q_0-1)^2\\
d \geq n-\ell q_0^2 - 2q_0^{j+1}(q_0-2)
\end{gather*}
and $r_1=(q_0-1)^2,r_2=q_0-1,$
  \begin{gather*}
\rho_1 \geq \max(2q_0-\deg_{\psi_{j+2}}(x)(q_0-2),4).
\end{gather*}
\end{proposition}
\begin{proof} Apply the construction of Proposition \ref{prop:general} to the curves in Eq.~\eqref{eq:gsmap}. The length of the obtained code equals the size of the evaluation set $D$, which is taken to be $|X(\ff_q)|-1-q_0^{j+1}$, accounting for removing the point at infinity as well as the $q_0^{j+1}$ ramified points above $0\in {\mathbb P}^1$ on $X.$ All the other parameters are found directly from Proposition \ref{prop:general}. \end{proof}

%Notice that the parameter $r_1$ in this construction is rigid because we take all the fibers of the map $\phi_{j+2}$ to form the set of the evaluation points of the middle code. 
%To add flexibility, it is possible to further partition this set of fibers
%into groups of size $s+1$ for $s\in\{1,\dots,q_0-1\}$. Then the length and the dimension of the middle code become $q_0(s+1)$ and
%$r_2s$, and the other parameters can be easily written out.

The shortcoming of the above construction is that it is unclear how to choose the primitive element $x$ such that $\deg_{\psi_{j+2}}(x)$ is small enough to guarantee a large value of the minimum distance of the middle code $\rho_1.$ 
It would be preferable if we could limit $\deg_{\psi_{j+2}}$ to 1 since this would force the middle code to be an optimal LRC code by itself. 

\subsection{H-LRC codes from power maps}\label{sec:power} To overcome the shortcomings of the previous construction, in this
 section we present a construction of H-LRC codes from the curves in the GS-tower for which the primitive element $x$ of the map constructed is naturally injective on the fibers of the map $\psi:X\to Z$ where $X=X_{j}$ is a GS curve and $Z$ is a 
 quotient curve that we are going to construct. Define the curve $X_{j,c}$ by its function field
   $$
    \kk(X_{j,c}) = \kk(x_1^c,x_2,\ldots,x_j),
    $$
where the variables $x_i$ are defined as above in \eqref{eq:gsdef}. Now let $a,b\ge 2$ be positive integers such that $(a+1)(b+1)|(q_0+1)$. 
Consider a tower of curves
  $$
  X_j\xrightarrow{\phi_2} X_{j,a+1}\xrightarrow{\phi_1} X_{j,(a+1)(b+1)}.
  $$
Applying the construction of Section \ref{sec:curves} with $x_1$ and $x_1^{a+1}$ as 
the primitive elements of $\phi_1$ and $\phi_2$ respectively, we obtain the following result, proved directly from
Proposition~\ref{prop:general}. We again rely on the notation $t=\dim(L(\ell Q_{\infty,j})),$ where $g_{j-1}\le \ell\le n_{j-1}$.
\begin{proposition}\label{prop:pow}
For any $j\ge 1$
there exists a family of H-LRC codes with parameters $[n,k,d]$ and locality $(r_1, \rho_1),(r_2,\rho_2 = 2)$, where
$n=q_0^{j-1}(q_0^2-1),k=tab$
\begin{gather*}
d\ge n-\deg(Q_{\infty})(a+1)(b+1)-q_0^{j-1}(ab+b-2)\nonumber\\
    r_2 = a,\hspace{5pt} \rho_2 = 2 \label{eq:r2}\\
    r_1 = ab,\hspace{5pt} \rho_1 = a+3.\nonumber
\end{gather*}
\end{proposition}
Note that the middle codes in this construction are optimal LRC codes, something that was not attainable with the construction of 
Prop.~\ref{prop:GS}. Further, taking $j=1$ in this proposition, we recover codes constructed of Prop.~\ref{prop:P1}, where
$n$ is taken to be $q-1.$

\begin{example}\label{ex:H} {\rm Let $q=q_0^2$ where $q_0$ is a prime power and let $X$ be the Hermitian plane curve of genus $g_0 = q_0(q_0-1)/2$ with the affine equation:
   $$
X: x^{q_0} + x = y^{q_0+1}.
   $$
Note that this curve coincides with the curve $X_2$ from the Garcia-Stichtenoth tower. The size of the evaluation set equals 
$q_0^3-q_0$ which corresponds to removing the $q_0$ points above $0\in{\mathbb{P}}^1$ on the curve. Applying the above power map construction to the case $q_0=8$ and $a=b=3$ gives a Hermitian H-LRC code 
defined over $\ff_{64}$. We obtain a family of codes with parameters $[n=504, k=9t, d]$ H-LRC code and locality $(9,6),(3,2)$ where:
  \begin{gather*}
    d\ge n-16t-80,\;
    1\le t \le 26.
\end{gather*}
In particular, we obtain codes with the following parameters:

  \vspace*{.1in}\begin{center}\begin{tabular}{ccc}
$t$ &$k$ &$d$ \\
1 &9&$408\le d\le 494$\\
2 &18&$392\le d\le 478$\\
3 &27&$376\le d\le 462$\\
&\dots\\
11&99&$248\le d\le334$\\
12&108&$232\le d\le318$
\\
&\dots
\end{tabular}
\end{center}
where the upper bound on $d$ is found from \eqref{eq:sh}. 
}
\end{example}
\vspace*{-.2in}
\rightline{\qed}

\subsection{H-LRC codes from fiber products}
{The result of Prop. \ref{prop:pow} affords a generalization based on fiber products of curves. Let us recall the definition of the
fiber product of curves $X$ and $Y$ over a curve $Z.$ Suppose that $\phi:X\to Z$ and $\psi:Y\to Z$ are $\kk$-covering maps. 
The set $X\times_Z Y:=\{(x,y)\in X \times Y|\phi(x)=\psi(y)\}$ is called a {\em fiber product} of $X$ and $Y$. In general this set
does not always form a smooth algebraic curve, but we will assume this in our discussion below. }

Consider a tower of projective smooth absolutely irreducible curves over a finite field $\kk$
  $$
X\xrightarrow{\phi_2}Y\xrightarrow{\phi_1}Z
$$
where as before $\deg(\phi_2)=ab$ and $\deg(\phi_1)=b$. 
Let us also assume that $\kk(X) = \kk(Z)(x)$ for some primitive element $x\in\kk(X)$ that is injective on fibers of $\phi_1 \circ \phi_2$. Choose a curve $C$ that forms a $\kk$-cover of $Z$ and such that $X\times_Z C$ and $Y\times_Z C$ are both smooth and absolutely irreducible curves. Then $x$ is injective on the fibers of $$X\times_Z C\to C(\cong Z\times_Z C)$$
Applying the construction of Section~\ref{sec:curves}, we obtain the following result.
\begin{proposition} Consider codes constructed using the tower
%Applying the construction of section \ref{sec:curves} to the tower
   $$
   X\times_Z C \xrightarrow{\phi_2}Y\times_Z C \xrightarrow{\phi_1} C.
   $$ 
   The parameters of the codes are $[n,k,d],$ where $n$ is determined by the number of totally split points on $X\times_Z C$
   and the distance $d$ satisfies the same condition as in Prop.~\ref{prop:pow}.
The locality parameters are  $(r_1, \rho_1),(r_2,\rho_2 = 2)$, where
\begin{gather*}
    r_2 = a,\hspace{5pt} \rho_2 = 2\\
    r_1 = ab,\hspace{5pt} \rho_1 = a+3.
\end{gather*}
The middle code has the length $(a+1)(b+1)$ and is an optimal LRC code with respect to the bound \eqref{eq:sb2}.
\end{proposition}
This construction specializes to Proposition \ref{prop:pow} with the choices $X, Y$ and $Z$ such that $\kk(X)=\kk(x_1)$, 
$\kk(Y)=\kk(x_1^{a+1})$ and $\kk(Z)=\kk(x_1^{(a+1)(b+1)}),$ and $C=X_{j,(a+1)(b+1)}, j\ge 1.$ 

Fiber products of Artin-Schreier curves, developed in \cite{vdGvdV95}, look especially promising for constructing H-LRC codes because they give curves with many points, including many maximal curves.

\subsection{H-LRC Codes with availability}
In this section we consider a generalization of codes with locality wherein local correction of erasures can be performed by
accessing several disjoint groups of codeword's coordinates. In the literature on LRC codes (without hierarchical structure) this
generalization is called the {\em availability problem} \cite{RPDV16}, \cite{tamobarg}, \cite{BTV17}, \cite{HMM}. We begin with the definitions and general expressions
of the parameters of the codes, and then give two examples, which form the main contents of this section.

Let us first define an LRC code with the availability property (and no hierarchy of recovering sets). The following definition is a slight
extension of the definition in \cite{RPDV16}.

\blue{\begin{definition} A linear code $\cC$ is LRC with locality $(r_j,\rho_j)_{1\le j\le\tau}$ and availability $\tau$ if
 for every $i\in[n]$ there are $\tau$ punctured codes $\cC_{i,1},\dots,\cC_{i,\tau}$ such that for every $j\in[\tau],$
   \begin{enumerate}
   \item $i\in\supp(\cC_{i,j}),$
   \item $\dim(\cC_{i,j})\le r_{j}$,
    \item $d(\cC_{i,j})\ge \rho_j$
     \item The set $\big|\text{\rm supp}(C_{i,j})\backslash\bigcup_{\begin{substack}{k\in[\tau]\\k\ne j}\end{substack}}\text{\rm supp}(C_{i,k})\big|$ contains $\dim(\cC_{i,j})$ linearly independent coordinates of the code $\cC_{i,j}.$
\end{enumerate}
\end{definition}
It may seem unnecessary to allow different parameters of the codes $\cC_{i,j},$ but the examples that we construct below are of this form. Moreover, we found it difficult to construct examples of codes from curves which do not use this generalization. The number of erasures
that can be corrected in parallel by the codes $C_{i,j}, j\in[\tau]$ equals $\min_{j}\rho_j -1.$
}

\vspace{.1in} Let us define H-LRC codes with availability. They generalize both LRC codes with locality $(r,\rho)$ and LRC codes with availability from the
cited works.

\blue{\begin{definition}[H-LRC codes with availability] \label{def:t-hie} Let $\tau_1,\tau_2\ge 1$ and let $\rho_{2,{j_2}}< \rho_{1,{j_1}}$ and $r_{2,{j_2}}\le r_{1,{j_1}}$ for $j_1\in[\tau_1], j_2\in[\tau_2].$ A linear code $\cC$ is H-LRC and 
parameters $((r_{1,j_1},\rho_{1,j_1}),(r_{2,j_2},\rho_{2,j_2}))$ and availability $\tau_1,\tau_2$ if 
\begin{enumerate}
\item it has locality $(r_{1,{j_1}},\rho_{1,{j_1}}),j_1\in[\tau_1]$ and availability $\tau_1$,
\item each of the codes $C_{i,j_1},i\in[n],j_1\in[\tau_1]$ is an LRC code with locality $(r_{2,j_2},\rho_{2,j_2}),j_2\in[\tau_2]$ and availability $\tau_2.$
\end{enumerate}
\end{definition}
}

%Consider the following scenario in the local recoverability problem: suppose two users wish to access the information at a failed node simultaneously. If the information was encoded using an LRC or H-LRC code then we normally would send those users to the repair coordinated for that node. This can cause issues for one of the users when they try to access the repair nodes and the servers are busy. To prevent this issue one looks to a subset of LRC codes with \textit{availability}. This means that every coordinate of the codes is equipped with multiple recovering sets that do not intersect. Such a code with $t$ recovering sets for every node is called an LRC(t) code.

{\em Remark} 
This definition can be specialized to the case when availability is required only for local recovery at the level of the entire code
$\cC$ (in this case $\tau_2=1$), or only at the level of the middle codes (in this case $\tau_1=1$).

%
%
%\[
%\rotatebox{45}{$
%\begin{matrix}
%\rotatebox{-45}{$(r,\rho)$ LRC } & \subseteq & \rotatebox{-45}{$B$} \\
%\rotatebox{-90}{$\subseteq$}& &\rotatebox{-90}{$\subseteq$}\\[9pt]
%\rotatebox{-45}{$C$} & \subseteq &\rotatebox{-45}{$ D$}
%\end{matrix}
%$}
%\]

\vspace*{.1in}
To generate H-LRC codes with availability we use a construction inspired by the LRC codes with availability introduced in \cite{BTV17} and developed in \cite{HMM}.

\vspace*{.1in}
 \begin{example}\label{example:FPHerm} {\rm \blue{Hermitian function fields, already mentioned above, provide an easy example of the fiber product construction of LRC codes with availability (see \cite{BTV17}, Sec.V.A, V.B). Let $\kk=\ff_q, q=q_0^2,$ and let $X,Y,Z$ be isomorphic to $\mathbb{P}^1_{\kk}.$ Let $X=\kk(x,y),$ where $x^{q_0}+x=y^{q_0+1}$, let $Y_1=\kk(x),Y_2=\kk(y),$
and let $Z=\kk(u),$ where $u=x^{q_0}+x, u=y^{q_0+1}$ as shown in the following diagram:}

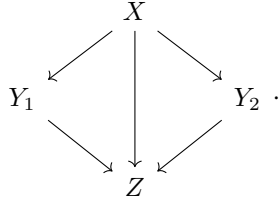
\begin{center}\begin{tikzcd}
%&X \arrow{rd}[description]{g_2}\arrow{ld}[description]{g_1}\arrow{dd}[description]{g}\\
%Y_1 \arrow{rd}[description]{h_1}&&Y_2 \arrow{ld}[description]{h_2}\\
%&Y
&X \arrow{rd}\arrow{ld}\arrow{dd}\\
Y_1 \arrow{rd}&&Y_2 \arrow{ld}\\
&Z
\end{tikzcd}.
\end{center}

 \blue{Then $X=Y_1\times_Z Y_2,$ and the function $y$ is constant on the fibers of the map $X\to Y_1$, while $x$ is constant on the fibers of the map
$X\to Y_2.$ This supports the univariate interpolation that underlies the local erasure recovery in LRC codes with availability.

}
}
\end{example}

 We will focus on the example where the availability on both levels of hierarchy is $\tau_1=\tau_2=2$ (even though 
it is possible to make it more general, already availability 2 results in cumbersome calculations, see the examples below).
Consider the diagram of curves given below where we assume that all the arrows correspond to separable maps between projective curves over a fixed finite field $\kk$.

\vspace*{.1in} Suppose that $Z$ is an absolutely irreducible smooth curve.
The curve $Y$ is constructed as the fiber product of two curves over $Z$ and the curve $X$ is constructed as the fiber 
product of two curves over $Y$. \blue{Suppose that there are $c$ points on $Z$ such that 
\begin{itemize}
\item[{(i)}] there are $t_1$ points of $Y_1$ above each of these points of $Z$;
\item[{(ii)}] there are $t_2$ points of $Y_2$ over each of these points of $Z$.
\end{itemize}
Now consider the points on $Y$ obtained as pairs of the points on $Y_1,Y_2$ described in (i)-(ii). Suppose that for each of these points
\begin{itemize}
\item[(iii)] there are $s_1$ points of $X_1$ above each of these points of $Y$;
\item[(iv)] there are $s_2$ points of $X_2$ above each of these points of $Y,$
\end{itemize}
where $\text{gcd}(s_1,s_2)=1$ and $\text{gcd}(t_1,t_2)=1.$
}

\begin{figure}[ht]
\begin{center}
\begin{tikzpicture}
  \matrix (m) [matrix of math nodes,row sep=3em,column sep=4em,minimum width=2em]
  {
     & X=X_1 \times_Y X_2 & \\
     X_1& &X_2 \\
     & Y= Y_1 \times_Z Y_2 &\\
     Y_1 & & Y_2\\
     & Z \\};
  \path[-stealth]
    (m-1-2) edge node [above left] {$\psi_{X_1}$} (m-2-1)
            edge node [above right] {$\psi_{X_2}$} (m-2-3)
    (m-2-1) edge node [below left] {$s_1$} (m-3-2)
    (m-2-3) edge node [below right] {$s_2$} (m-3-2)
    (m-3-2) edge node [above left] {$\psi_{Y_1}$} (m-4-1)
            edge node [above right] {$\psi_{Y_2}$} (m-4-3)
    (m-4-1) edge node [below left] {$t_1$} (m-5-2)
    (m-4-3) edge node [below right] {$t_2$} (m-5-2);
\end{tikzpicture}
\end{center}
\caption{Diagram of covering maps for the construction of codes with availability}\label{fig:avail}
\end{figure}
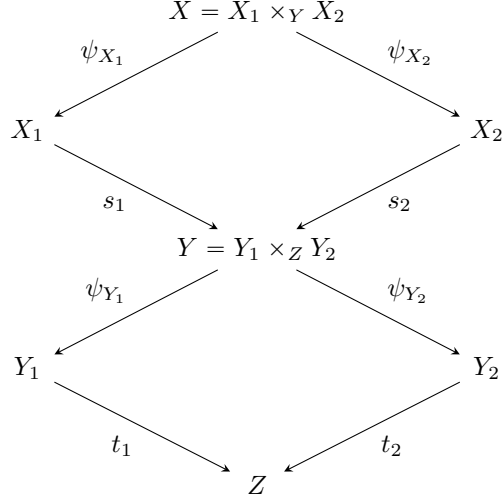

Let $D=\{(Q_1,Q_2)\}\subset X,$ where $Q_1$ runs over the points of $X_1$ constructed in (iii) and $Q_2$ over the points of $X_2$ constructed in (iv). Note that the size of the set $D$
is $n=cs_1s_2t_1t_2.$ Choose a positive divisor $Q_{\infty}$ on $Z$ with $L(Q_{\infty})=
 \text{span}\,\{f_1,\ldots,f_m\}$ and choose primitive elements $x_1, x_2, y_1,$ and $y_2$ such that 
 $\kk(X_i)=\kk(Y)(x_i)$ and $\kk(Y_i)=\kk(Z)(y_i),$ $i=1,2.$ 
\blue{ Assume that the degrees of $x_1,x_2$ considered as maps from $X$ to 
$\mathbb{P}^1$ are $h_{x_1}:=\deg(x_1), h_{x_2}:=\deg(x_2)$ and the degrees of $y_1,y_2$ as maps from $Y$ to $\mathbb{P}^1$ are $h_{y_1},h_{y_2}.$ 
%Further, let $h'_{ij}$ be the \blue{largest number of zeros} that
%$x_i$ can have on a fiber of the map $X\to Y_j,i,j=1,2$ (this is similar to $\deg_{\psi}(x)$ defined before Proposition~\ref{prop:general}).
Further, let   $h'_{i,j}$ be the maximum possible number of zeros of $x_i$ on a fiber of the map $X\to Y_j , i, j =1, 2.$ (this is similar to $
\deg_{\psi}(x)$ defined before Proposition~\ref{prop:general}). Therefore, $h'_{ii}=s_i, i=1, 2$ and, if the fibers of these maps are transversal, then also $h'_{12}=h'_{21}=1$.}
 Let $V$ be the space of functions given by
   \begin{align*}
   V=\operatorname{span}\{f_ix_1^{j_1}x_2^{j_2}y_1^{k_1}y_2^{k_2}\,|\,i=1,\ldots,m;
     j_l=0,\dots,s_l-2, k_l=0,\dots,t_l-2;l=1,2\}
     \end{align*}
Define the code $\mathcal{C}$ as the image of the evaluation map
\begin{align*}
    \operatorname{ev}\!:\,&V\to \kk^n\\
    &v\mapsto(v(P_i)|P_i\in D).
\end{align*}
\blue{The code $\cC$ is supported on all the $n$ points in $D.$ 
Each of the middle codes $\cC_{i,1}$ ($\cC_{i,2}$) is supported on the fibers of the map $X\to Y_1$ (resp., $X\to Y_2$). 
The length of the codes is $\nu_1=s_1s_2t_2$ and $\nu_2=s_1s_2t_1,$ respectively. The bases of function spaces that give the codes
$\cC_{i,j},j=1,2$ are
  \begin{gather*}
  V_1=\{x_1^{j_1}x_2^{j_2}y_1^{k_1}\mid j_l=0,\dots,s_l-2,l=1,2; k_1=0,\dots,t_1-2\}\\
  V_2=\{x_1^{j_1}x_2^{j_2}y_2^{k_2}\mid j_l=0,\dots,s_l-2,l=1,2; k_2=0,\dots,t_2-2\}.
  \end{gather*}
  These codes are LRC with repair groups of size $s_2$ for $j=1$ and $s_2$ for $j=2.$ Local correction of a single erasure 
can be performed in parallel along the corresponding fibers of the maps
$X\to X_j,j=1,2$ (cf. Fig.~\ref{fig:cover}).}

The properties of the code $\cC$ are collected in the following proposition.
\blue{
\begin{proposition}\label{prop:HLRC} The code $\mathcal{C}$ is an $[n,k,d]$ H-LRC code with parameters
\begin{gather*}
    n=cs_1s_2t_1t_2\\
    k=m(s_1-1)(s_2-1)(t_2-1)(t_2-1)\\
    d\ge n-\deg (Q_\infty)s_1s_2t_1t_2 - (h_{y_1}(t_1-2)+h_{y_2}(t_2-2))s_1s_2-h_{x_1}(s_1-2)-h_{x_2}(s_2-2),
\end{gather*}
availability $2$, and locality $(r_{11},\rho_{11}),(r_{12},\rho_{12}),$ where
\begin{gather*}
       r_{11} = (s_1-1)(s_2-1)(t_1-1) \nonumber\\
%       \rho_1 \ge \max\{(r_2+1)^2(s+1) - (s-1)(r_2+1)^2 - 2h'_x(r_2-1),8\}. \label{eq:Hrho1}
      \rho_{11}\ge \max(s_1s_2(t_2-t_1+2)-h'_{11}(s_1-2)-h'_{21}(s_2-2),4)\\
      r_{12}=(s_1-1)(s_2-1)(t_2-1)\\
      \rho_{12} \ge \max(s_1s_2(t_1-t_2+2)-h'_{12}(s_1-2)-h'_{22}(s_2-2),4).
\end{gather*}
The middle codes $\cC_{i,1},\cC_{i,2}$ are H-LRC(2) codes % in this construction is obtained by restricting $\mathcal{C}$ 
%to a fiber of either of $X\to Y_i,i=1,2$ and is itself an H-LRC(2) code 
with locality parameters $(s_2-1,2)$ and $(s_1-1,2),$ respectively.
\end{proposition}
\begin{IEEEproof}
The parameters of the codes follow directly from the construction. Specifically, the code length is obtained from the count of points and
the dimensions are found by counting the size of the corresponding functional bases in $V,V_1,V_2.$ The estimates of the distances of the codes $d,\rho_{11},\rho_{12}$ are found from the bounds on the largest possible number of zeros of the functions in the bases. For instance, 
  $$
  \rho_{11}\ge s_1s_2t_2-(s_1s_2(t_1-2)+h'_{11}(s_1-2)+h'_{21}(s_2-2)),
  $$
where the terms in the parenteses bound above the maximum number of zeros of a function in $V_1.$
%first term is the length of $\cC_{i,1}$, the second equals the number of points over a point on $Y$ times that maximum power of $y_1$,
%the third is the number of zeros of $x_1$ on the fiber of the map $X\to Y$ times the maximum power of $x_1$, and the fourth is similar.
  
   The lower bounds $\rho_{1,\ast}\ge 4$ in the estimates 
of $\rho_{11},\rho_{12}$ is justified exactly as in Prop.~\ref{prop:general}.
%The lower bound of $8$ on $\rho_1$ in \eqref{eq:Hrho1}
%follows by the argument similar to the proof of . 
\end{IEEEproof}}
We note that for the bounds on $\rho_{11}$ and $\rho_{12}$ to be simultaneously greater than 4, is is necessary that $|t_1-t_2|\le 1.$ This is
somewhat restrictive, but still possible to account for in examples, see the next section.

\subsection{Families of H-LRC codes with availability}\label{sec:H-LRC-A}
While the description in the previous section relies on the bottom-up construction that starts with the curve $Z$, examples are easier to obtain
using a top-down approach.

\subsubsection{Construction from RS codes} \blue{In this section we construct a family of H-LRC codes, extending the construction of 
Sec.~\ref{sec:RS} above. Let $\ff_q$ be a finite field with generating element $g,$ and suppose that $q-1=cs_1s_2t_1t_2,$ 
where $\text{gcd}(s_1,s_2)=1,\text{gcd}(t_1,t_2)=1$. We will take $D=\ff_q^\ast$ and construct codes of length $n=q-1.$
Consider the following subgroups of the cyclic group $\ff_q^\ast:$
   $$
   \begin{aligned}
   G_1=\langle g^{t_1l}\rangle,\quad H_1=\langle g^{s_2t_1 t_2 l}\rangle\\
   G_2=\langle g^{t_2 l}\rangle,\quad H_2=\langle g^{s_1t_1 t_2 l}\rangle.
   \end{aligned}
  $$
The groups $H_1$ and $H_2$ define a pair of mutually orthogonal partitions of the set $[n]$ (meaning that the blocks of the partitions
intersect on at most one point). 

Referring to Fig.~\ref{fig:avail}, we obtain a diagram of curves (defined by their function fields) as follows:

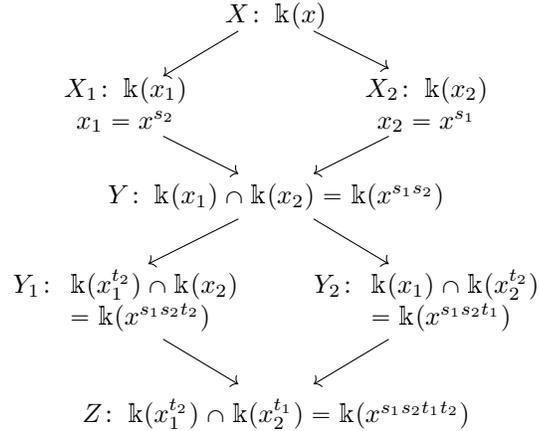
\begin{center}\begin{tikzpicture}
\draw(3,1.5) node[] {$X\!:\; \kk(x)$};
\draw(1,0.3) node[] {$\begin{array}{c}X_1\!:\; \kk(x_1)\\ x_1=x^{s_2}\end{array}$};
\draw(5,0.3) node[] {$\begin{array}{c}X_2\!:\; \kk(x_2)\\ x_2=x^{s_1}\end{array}$};
\draw(3,-0.9) node[] {$Y\!:\; \kk(x_1)\cap \kk(x_2)=\kk(x^{s_1s_2})$};
\draw(1,-2.3) node[] {\begin{tabular}{l@{\hspace*{0.1in}}l}$Y_1\!:$ &$\kk(x_1^{t_2})\cap\kk(x_2)$\\
                     &$=\kk(x^{s_1s_2t_2})$\end{tabular}};
\draw(5,-2.3) node[] {\begin{tabular}{l@{\hspace*{0.1in}}l}$Y_2\!:$ &$\kk(x_1)\cap\kk(x_2^{t_2})$\\
                     &$=\kk(x^{s_1s_2t_1})$\end{tabular}};
\draw(3,-3.8) node[] {$Z\!:\; \kk(x_1^{t_2})\cap\kk(x_2^{t_1})=\kk(x^{s_1s_2t_1t_2})$};

\draw[thin,->](2.5,1.3) -- (1.5,0.7) node[pos=0.4, anchor=north] {};
\draw[thin,->](3.5,1.3) -- (4.5,0.8) node[pos=0.4, anchor=north] {};

\draw[thin,->](4.5,-.1) -- (3.5,-0.6) node[pos=0.4, anchor=north] {};
\draw[thin,->](1.5,-.1) -- (2.5,-0.6) node[pos=0.4, anchor=north] {};

\draw[thin,->](3.5,-1.2) -- (4.5,-1.8) node[pos=0.4, anchor=north] {};
\draw[thin,->](2.5,-1.2) -- (1.3,-1.8) node[pos=0.4, anchor=north] {};

\draw[thin,->](1.5,-2.8) -- (2.5,-3.4) node[pos=0.4, anchor=north] {};
\draw[thin,->](4.5,-2.8) -- (3.5,-3.4) node[pos=0.4, anchor=north] {};

\end{tikzpicture}
\end{center}

The covering maps $\Psi$ are defined in an obvious way, and are of degrees $\deg(\Psi_{X_l})=s_l$ and $\deg(\Psi_{Y_l})=t_l,l=1,2.$

%  \begin{array}{c@{\hspace*{-.05in}}c@{\hspace*{-.05in}}c}
%  &X\!:\; \kk(x)&\\
%  \begin{array}{c}X_1\!:\; \kk(x_1)\\ x_1=x^{s_2}\end{array} &&\begin{array}{c}X_2\!:\; \kk(x_2)\\ x_2=x^{s_1}\end{array}\\
%  &Y\!:\; \kk(x_1,x_2)&\\[.05in]
%  Y_1\!:\; \kk(x_1^{t_2},x_2)&& Y_2\!:\; \kk(x_1,x_2^{t_1})\\[.05in]
%  &Z\!:\; \kk(x_1^{t_1},x_2^{t_2})&
%  \end{array}

The codewords are obtained by evaluating functions of the following form:
     $$
     v=\sum_{\bar i} v_{\bar i} f_i x_1^{j_1}x_2^{j_2}y_1^{k_1}y_2^{k_2},
     $$
     where $\bar i=(i,j_1,j_2,k_1,k_2), y_1=x_2^{s_2t_2},y_1=x_1^{s_1t_1}$, and the summation runs over the range
\begin{gather*}
1\le i\le m\\
 0\le j_l\le s_l-2,
 0\le k_l\le t_l-2, \; l=1,2.
\end{gather*}
The value of $m=\deg(Q_\infty)+1$ is a parameter of the construction, and it is chosen so that the estimates of the other parameters do not trivialize. Once
it is fixed, the values of $k,d$ and locality are obtained directly from Proposition \ref{prop:HLRC}, where we have 
$h_{y_1}=t_2,h_{y_2}=t_1,h_{x_1}=s_2t_1t_2,h_{x_2}=s_1t_1t_2,$ and $h'_{11}=s_1,h'_{22}=s_2,h'_{12}=h'_{21}=1.$ In particular,
   \begin{gather}
   d\ge (c-m-3)s_1s_2t_1t_2+2(s_1s_2(t_1+t_2)+t_1t_2(s_1+s_2))\nonumber\\
   \rho_{11}\ge s_1s_2(t_2-t_1+2)-s_1(s_1-2)-(s_2-2) \label{eq:rho11}\\
   \rho_{12}\ge s_1s_2(t_1-t_2+2)-(s_1-2)-s_2(s_2-2) \label{eq:rho12}
   \end{gather}
     To give a numerical example, let $q=41^2,$ then we can take $s_1=7,s_2=3,t_1=4,t_2=5,c=4.$ Taking $m=1,$ we obtain $k=144,d\ge 778$
and for the middle codes the parameters $[\nu_1=105,r_{11}=36,\rho_{11}\ge 26]$ and $[\nu_2=84,r_{12}=48,\rho_{12}\ge 13].$

In conclusion we note that a one-level version of this construction (LRC codes with two disjoint recovering sets, but with no
hierarchical structure) was given in \cite[Sec. IV]{tamobarg}.

}

\vspace*{.1in}
\subsubsection{Construction from Hermitian curves} \blue{Let $X$ be a Hermitian curve over $\kk=\ff_q, q=q_0^2,$ given by the affine equation
   $$
   X\!:\; x^{q_0}-x=y^{q_0+1}.
   $$
As in Example~\ref{example:FPHerm}, we view $X$ as a fiber product. To implement the construction in Fig.~\ref{fig:avail}, we need the following data. Consider a natural projection $X(\ff_q)\to {\mathbb P}^1$ given by
$(x,y)\mapsto x.$ Let $M:=\{x\in\ff_{q}: x^{q_0}+x=0\}$ and take the point set $D=\ff_q\backslash M.$ The set $D$ will be the
evaluation set of points of the constructed code, and thus, $n=|D|=q_0^3-q_0.$ 
Let $q_0=s_2^2,$ where $s_2=p^3, t_2=p^2,p=\text{char}\, \kk$ and let $q_0+1=s_1t_1 c'.$ Let $n=cs_1s_2t_1t_2,$ where
  $$
  c=\frac{n}{s_1s_2 t_1t_2}=c'\frac{q_0(q_0^2-1)}{(q_0+1)s_2t_2}=c'p(q_0-1).
  $$
Overall with this choise of the parameters, we obtain $n=q_0(q_0^2-1).$

As above, assume that $\operatorname{gcd}(s_1,s_2)=
\operatorname{gcd}(t_1,t_2)=1.$ Define
   $$\begin{array}{c}\begin{matrix}
   v_1=x^{s_2}-x, & v_2=y^{s_1}\\
   u_1=v_1^p+v_1, & u_2=v_2^{t_1}.
   \end{matrix}
   \end{array}
   $$
Let    
   $$
  \begin{aligned}
  X_1\!:\; &v_1^{s_2}+v_1=y^{q_0+1}  \\
      &\kk(X_1)=\kk(v_1,y)\\[.03in]
  X_2\!:\; &x^{q_0}-x=v_2^{t_1c'}\\
      &\kk(X_2)=\kk(x,v_2).
 \end{aligned}
   $$ 
We have $X=X_1 \times_Y X_2$ and $\kk(X)=\kk(x,y)=\kk(X_1)(x)$ and $\kk(X)=\kk(X_2)(y),$ where the curve $Y$ with the function
field $\kk(Y)=\kk(v_1,v_2)$ is defined as
  $$
  Y: v_1^{s_2}+v_1=v_2^{t_1c'}
  $$
This curve is a fiber product of the curves $Y_1,Y_2$ over $Z$, where 
  $$
  \begin{aligned}  
     &Y_1\!:\; &&u_1^{p^2}-u_1^p+u_1=v_2^{t_1c'}\\
        &&&\kk(Y_1)=\kk(u_1,v_2)\\[.03in]
     &Y_2\!:\; && v_1^{s_2}+v_1=u_2^{c'}\\
        &&&\kk(Y_2)=\kk(u_2,v_1)\\
     &Z\!:\; &&u_1^{p^2}-u_1^p+u_1=u_2^{c'}\\
             &&&\kk(Z)=\kk(u_1,u_2).
    \end{aligned}
    $$    
This completes the desired commutative diagram. %Note that $c$ is the number of points of $Z$ used in the construction, and it is a parameter to be chosen. 
The primitive elements of the field extensions are $(x_1,x_2,y_1,y_2)=(y,x,v_2,v_1)$ and their degrees are given by
 \begin{gather*}
  h_{x_1}=q_0+1,\; h_{x_2}=q_0,\;h_{y_1}=t_2,\; h_{y_2}=t_1\\
  h_{11}'=s_1,\;h_{22}'=s_2,\;h_{12}'=h_{21}'=1.
  \end{gather*}

The expressions for the code parameters are obtained directly from Propostion~\ref{prop:HLRC}, including \eqref{eq:rho11}-\eqref{eq:rho12}. 
To give an example, let $q_0=64, p=2, c'=1,s_1=8,s_2=13,t_1=5,t_2=4.$ The parameters of the code $\cC$ depend on the choice of $\deg(Q_\infty)$
and can vary. The parameters of the middle codes are $\nu_1=416,r_{11}=336,\nu_2=520,r_{12}=252.$ From \eqref{eq:rho12} we obtain $\rho_{12}\ge 163,$ and for $\rho_{11}$ we can only claim the lower bound of $4$ (the true distance is likely higher).

Observe that this construction affords many versions, and we give one of the simplest possible of them.

\vspace*{.1in}In conclusion we note that it is possible to extend this example to the tower of Garcia-Stichtenoth curves, obtaining asymptotically good sequences of H-LRC codes with availability. %It is not clear if including the calculations is warranted.  
}

\section{Asymptotic parameters}
In this section we consider asymptotic parameters of H-LRC codes. In the setting that we adopt, the code length $n\to\infty,$ and we
call the codes asymptotically good if the limits of the rate $R:=(1/n)\log_q|\cC|$ and relative distance $\delta:=d/n$ both are bounded away from 0 as $n\to\infty.$ The parameters of the middle code $[\nu,r_1,\rho_1]$ are constant and do not depend on $n$.

\subsection{Asymptotically good families of H-LRC codes}
Let us compute the asymptotics of the code parameters in Prop.~\ref{prop:GS}. Recall that $g_j\le\frac{n_j}{q_0-1}$ \cite{garcia95}.
We have
  \begin{align}
  \frac{d}{n}+\frac{k}{n}\frac{q_0^2}{(q_0-1)^2}&= 1-\frac{2(q_0-2)}{q_0^2-1}-\frac{q_0^2g_j}{n}+\frac{q_0^2}{n}\nonumber\\
  &\ge 1-\frac{3}{q_0+1}+\frac{q_0^2}{n}.\label{eq:SG}
  \end{align}
  
  We obtain the following code family.  
  \begin{proposition}\label{prop:asymp} Let $q=q_0^2.$  %and suppose that $d/n\to\delta$ and $k/n\to R$
There exists a family of linear $q$-ary 2-level H-LRC codes with locality $(((q_0-1)^2,\rho_1),(q_0-1,2)),$ where $\rho_1$ satisfies the
bound of Proposition \ref{prop:GS}, and such that the rate and relative distance satisfy the inequality
  \begin{equation}\label{eq:ab}
  R\ge \Big(\frac{q_0-1}{q_0}\Big)^2\Big(1-\delta-\frac{3}{q_0+1}\Big).
  \end{equation}
\end{proposition}
The bound \eqref{eq:ab} is obtained by letting $j\to\infty$ and passing to the limit in \eqref{eq:SG}.

To add flexibility to the parameters of the code family, we can decrease the maximum degrees of $x,y$ in the functions in \eqref{eq:V} from 
$s-1$ to $s_1-1$ and from $r_2-1$ to $r_2'-1$, where $2\le s',r'_2\le q_0-1.$ This gives the following extension of Proposition~\ref{prop:asymp}.
\begin{proposition}\label{prop:asymp1}
There exists a family of linear $q$-ary 2-level H-LRC codes with locality $$((r_1=r_2s,\rho_1),(r_2,\rho_2=q_0+1-r_2)),\quad 2\le s,r_2\le q_0-1$$ and
   $$
   R\ge\frac{sr_2}{q_0^2}\Big(1-\delta-\frac{q_0+s+r_2-1}{q_0^2-1}\Big).
   $$
\end{proposition}

Observe that, while the code families in the previous two propositions are asymptotically good, the distance of the middle codes
$\rho_1$ does not have an explicit expression. This can be remedied by using the code family of Proposition \ref{prop:pow}, and performing a calculation similar to \eqref{eq:SG}. We obtain the following theorem which gives a fully explicit set of parameters for 
an asymptotically good family of H-LRC codes.
\begin{theorem} Let $q=q_0^2$  and suppose that $\nu:=(a+1)(b+1)|(q_0+1).$
There exists a family of linear $q$-ary 2-level H-LRC codes with locality 
  $(r_1=ab,\rho_1=a+3),(r_2=a,\rho_2=2))$ and the rate and relative distance satisfying the asymptotic bound
  \begin{equation}\label{eq:pa}
  R\ge\frac{ab}{(a+1)(b+1)}\Big(1-\delta-\frac{q_0+ab+b-1}{q_0^2-1}\Big).
  \end{equation}
  The $[\nu,r_1,\rho_1]$ middle codes in the construction are distance-optimal in that they satisfy the bound \eqref{eq:sb} with equality.
\end{theorem}
\begin{proof}
From Proposition \ref{prop:pow} we obtain:
\begin{align}
\frac dn+\frac kn\frac{(a+1)(b+1)}{ab}\ge 1-\frac{ab+b-2}{q_0^2-1}-\frac{(g_{Z}-1)(a+1)(b+1)}n \label{eq:dn}
\end{align}
where $g_Z$ is the genus of the curve $X_{j,(a+1)(b+1)}.$ Recalling the Riemann-Hurwitz formula \cite[p.102]{TVN07}, we obtain the 
relation $g_j\ge1+(a+1)(b+1)(g_Z-1),$ which gives
  $$
  \frac{(g_{Z}-1)(a+1)(b+1)}n\le \frac{g_j-1}n
  $$
Substituting this in \eqref{eq:dn}, we continue as follows:
 \begin{align*}
\frac dn+\frac kn\frac{(a+1)(b+1)}{ab}&\ge 1-\frac{ab+b-2}{q_0^2-1}-\frac{g_j-1}n
 \end{align*}
 Since $\frac {g_j}n\to \frac 1{q_0-1},$ we obtain \eqref{eq:pa} upon rearranging.
\end{proof}

To get an idea of the bound \eqref{eq:pa}, assume that $a=b\approx q_0.$ Assuming large $q_0$ and ignoring small terms, we find
that the right-hand side of \eqref{eq:pa} is approximately $1-\delta-\frac{2}{\sqrt{q}}$ and is in fact better than the bound 
\eqref{eq:ab}.

\subsection{A random coding argument}
As in \cite{BTV17}, let us also compute a bound on the set of achievable pairs $(R,\delta)$ obtained by a random coding argument, calling
it a Gilbert-Varshamov (GV) type bound.
Consider a sequence of $q$-ary H-LRC codes $\cC^{(i)}$ of length
$n_i$ with locality $((r_1,\rho_1),(r_2,\rho_2))$. Suppose that $d_i$ is the distance of the code $\cC^{(i)}$ and let $\frac{d_i}{n_i}\to \delta$ as $i\to\infty.$

% Let $M_q(n,\bfr,\delta n)=\max |\cC|,$ where $\cC$ is a $q$-ary 2-level H-LRC code with the described
%properties. Let
%    $$
%      R_q(\bfr,\delta)=\limsup_{n\to\infty} \frac 1n\log_q M_q(n,\bfr,\delta n).
%    $$
%    Note that the asymptotics that we consider lets the distance to scale with $n$ while keeping fixed all the parameters associated with local recovery. 
    % (in a different form the same ideas appear in \cite{CaMa2015}).
    \begin{proposition}\label{prop:GV2}{\sc(GV bound)} Assume that there exists a $q$-ary $[\nu,r_1,\rho_1]$ linear LRC code $\cD$ with locality $(r_2,\rho_2)$ and let $B_\cD(s)$ be the weight enumerator of the code $\cD$. For any $R>0,\delta>0$ that satisfy the inequality
\begin{equation}\label{eq:GV2}
  R< \frac {r_1}{\nu}-\min_{s>0}\Big(\frac 1{\nu}\log_q B_\cD(s)-\delta\log_qs\Big),
\end{equation}
there exists a sequence of H-LRC codes with asymptotic rate $R$ and relative distance $\delta.$
\end{proposition}
\begin{proof} The ideas in the following calculation extend the approach to a Gilbert-Varshamov bound for LRC codes derived in \cite{TamoBargEtAl2016,BTV17}, so we only outline the argument. Let $\cC$ be an $[n,k=Rn,d=\delta n]$ linear H-LRC code with locality parameters $\bfr=((r_1,\rho_1),(r_2,\rho_2))$ as given in Def.~\ref{def:hie}. Its parity-check matrix can be taken in the form $H=(H_1|H_0)^T,$ where the submatrices are as follows. The part $H_1$ is a block-diagonal matrix with blocks given by the parity-check matrix of the code $\cD.$ The matrix $H_0$ is formed of random uniform independent
elements of the field $\ff_q$ chosen independently of each other. The matrix $H_1$ contains $n({\nu}-{r_1})/{\nu}$ rows and
the matrix $H_0$ contains $n\frac {r_1}{\nu}-k$ rows.

The number of vectors of weight $w=1,\dots,n$ in the null space of $H_1$ is given by $\min_{s>0}s^{-w}B_\cD(s)^{n/{\nu}}$, and the probability that each of them is also in the null
space of $H_0$ is $q^{-n(\frac {r_1}{\nu}-R)}.$ By the union bound, 
   $$
   P(d_{\text{min}}(\cC)\le\delta n)=\delta n q^{-n(\frac {r_1}{\nu}-R)}\min_{s>0}s^{-w}B_\cD(s)^{n/{\nu}}.
   $$
If this probability is less than one, there exist codes with distance $d_{\text{min}}\ge \delta n.$ Upon taking logarithms, we now obtain \eqref{eq:GV2}.
\end{proof}
%Even though the distance of the middle node $\cD$ does not appear explicitly in the bound \eqref{eq:GV2}, the bound still depends on it
%via the weight 

Numerical comparison of the bounds obtained above, including \eqref{eq:pa} and \eqref{eq:ab}, with the GV bound is difficult because
\eqref{eq:GV2} is not easy to compute. Indeed, we need to find the weight distribution of the code $\cD$ (for instance, a
code in the family constructed in \cite{tamobarg}, see Sec.~\ref{sect:TBcodes}); however this is not easy even for moderate values of $q_0$. 
It is possible to replace \eqref{eq:GV2} with a weaker bound by observing that the codes of \cite{tamobarg} are subcodes of certain
Reed-Solomon codes (more specifically, a $q$-ary $[\nu,r_1,\rho_1]$ code $\cD$ is a subcode of the $[\nu,\nu-\rho_1+1,\rho_1]$ RS code), and  therefore, their
weight distributions are bounded above by the weight distribution of RS codes for which an explicit expression is available. 
Thus, we can use this expression to evaluate a lower estimate for the right-hand side of \eqref{eq:GV2}.
Following this route, we have computed numerical examples, observing that \eqref{eq:pa} indeed improves upon this version of the GV bound. One such example is as follows.

Let $q_0=19,a=3,b=4,$ then $\nu=20,r_1=12,\rho_1=6.$ Using the weight numerator of the $[20,15,6]$ RS code over $\ff_{19^2}$ on the right-hand side of \eqref{eq:GV2}, we find that rate $R=0.198$ is attainable for the relative distance $\delta=0.5.$ For the same $\delta$ the bound \eqref{eq:pa} produces a higher value $R\ge 0.243.$ 

We note again that this example does not imply that the bound \eqref{eq:pa} improves upon the actual GV bound which even for the above parameters is not easily computable.

%{\em Remark:} To compute the bound \eqref{eq:GV2} numerically, we can take the code $\cD$ from
%the family constructed in \cite{tamobarg}. For the chosen values of $r_1,r_2,\rho_2$ and ${\nu}$ these codes have the largest possible distance $\rho_1.$  %However, direct calculations with these codes are difficult because their weight distribution has to be found by exhaustive search. 
%To simplify the calculations, it is possible to use the fact that these codes are subcodes of certain Reed-Solomon codes, and
%%Since the GV bound \eqref{eq:GV2} is convex and \eqref{eq:
%already in \cite{BTV17} examples are computed over the field $\ff_{23^2},$ which is too large for most explicit calculations. At the same time, it is possible to compute a slightly
%weaker bound using the fact that the code $\cD$ is a subcode of the Reed-Solomon code $\cD'$ of length ${\nu}$ and dimension $r_1+\frac{r_1}{r_2}-1.$ Therefore, the number of codewords of any weight $w$ in $\cD'$ is greater or equal than in $\cD,$ and also $B_{\cD'}(s)\ge B_\cD(s)$ for all $s\ge 0.$ For this reason, we can replace $B_\cD$ with $B_{\cD'}$ in Prop.~\ref{prop:GV2},
%still getting a valid bound in \eqref{eq:GV2}.
%

\bibliographystyle{IEEEtranS}
\bibliography{main}
\end{document}